\pdfoutput=1
\documentclass[11pt]{article}
\usepackage{epsfig}
\usepackage{amsthm}
\usepackage{amsmath}
\usepackage{amsfonts}
\usepackage{url}
\usepackage{subfigure}

\title{\textbf{Staged Self-Assembly:} \\
  Nanomanufacture of Arbitrary Shapes with $O(1)$ Glues}
\author{
  Erik D. Demaine%
    \thanks{MIT Computer Science and Artificial Intelligence Laboratory,
      32 Vassar St., Cambridge, MA 02139, USA,
      \protect\url{{edemaine,mdemaine}@mit.edu}.
      Partially supported by NSF CAREER award CCF-0347776
      and DOE grant DE-FG02-04ER25647.}
\and
  Martin L. Demaine\footnotemark[1]
\and
  S\'andor P. Fekete%
    \thanks{Institut f\"ur Mathematische Optimierung,
      Technische Universit\"at Braunschweig,
      Pockelsstr.~14, 38106 Braunschweig, Germany,
      \protect\url{s.fekete@tu-bs.de}}
\and
  Mashhood Ishaque%
    \thanks{Department of Computer Science, Tufts University,
      Medford, MA 02155, USA, \protect\url{{mishaq01,dls}@cs.tufts.edu}.
      Partially supported by NSF grant CCF-0431027.}
\and
  Eynat Rafalin%
    \thanks{Google Inc., \protect\url{erafalin@cs.tufts.edu}.
      Work performed while at Tufts University.
      Partially supported by NSF grant CCF-0431027.}
\and
  Robert T. Schweller%
    \thanks{Department of Computer Science, University of Texas-Pan American, Edinburg, TX 78539, USA,
      \protect\url{schwellerr@cs.panam.edu}}
\and
  Diane L. Souvaine%
    \footnotemark[3]
}
\date{}

\newif\ifabstract
\abstractfalse
\newif\iffull
\ifabstract \fullfalse \else \fulltrue \fi

\usepackage
  [breaklinks,bookmarks,bookmarksnumbered,bookmarksopen,bookmarksopenlevel=2]
  {hyperref}
{\makeatletter \hypersetup{pdftitle={\@title}}}

{\makeatletter
 \gdef\xxxmark{%
   \expandafter\ifx\csname @mpargs\endcsname\relax 
     \expandafter\ifx\csname @captype\endcsname\relax 
       \marginpar{xxx}
     \else
       xxx 
     \fi
   \else
     xxx 
   \fi}
 \gdef\xxx{\@ifnextchar[\xxx@lab\xxx@nolab}
 \long\gdef\xxx@lab[#1]#2{{\bf [\xxxmark #2 ---{\sc #1}]}}
 \long\gdef\xxx@nolab#1{{\bf [\xxxmark #1]}}
 \long\gdef\xxx@lab[#1]#2{}\long\gdef\xxx@nolab#1{}%
}
\advance \topmargin by -\headheight
\advance \topmargin by -\headsep
\newtheorem{theorem}{Theorem}

\newtheorem{corollary}{Corollary}

\let\realbfseries=\bfseries
\def\bfseries{\realbfseries\boldmath}

\def\compactify{\itemsep=0pt \topsep=0pt \partopsep=0pt \parsep=0pt}
\let\latexusecounter=\usecounter
\newenvironment{itemize*}
  {\def\usecounter{\compactify\latexusecounter}
   \begin{itemize}}
  {\end{itemize}\let\usecounter=\latexusecounter}
\newenvironment{enumerate*}
  {\def\usecounter{\compactify\latexusecounter}
   \begin{enumerate}}
  {\end{enumerate}\let\usecounter=\latexusecounter}
\newenvironment{description*}
  {\begin{description}\compactify}
  {\end{description}}

\def\captionfont{\small\rm}
\def\captionlabelfont{\small\bf}
{\makeatletter
 \global\let\plainfont@makecaption=\@makecaption
 \long\gdef\@makecaption#1#2{%
   \plainfont@makecaption{\captionlabelfont #1}{\captionfont #2}}}

\def\nullglue{{\rm null}}
\def\emptytile{{\rm empty}}
\def\north{\mathop{\rm north}\nolimits}
\def\south{\mathop{\rm south}\nolimits}
\def\west{\mathop{\rm west}\nolimits}
\def\east{\mathop{\rm east}\nolimits}
\def\diameter{\mathop{\rm diameter}\nolimits}

\begin{document}
\maketitle

\vspace{-.25in}

\begin{abstract}
We introduce \emph{staged self-assembly} of Wang tiles,
where tiles can be added dynamically in sequence and
where intermediate constructions can be stored for later mixing.
This model and its various constraints and performance measures
are motivated by a practical nanofabrication scenario through
protein-based bioengineering.
Staging allows us to break through the traditional lower bounds in
tile self-assembly by encoding the shape in the staging algorithm
instead of the tiles.  All of our results are based on the practical
assumption that only a constant number of glues, and thus only a constant
number of tiles, can be engineered. Under this assumption, traditional tile self-assembly cannot even manufacture
an $n \times n$ square; in contrast, we show how staged assembly in theory enables
manufacture of arbitrary shapes
in a variety of precise formulations of the model.  
\end{abstract}




\section{Introduction}
\label{sec:intro}

\emph{Self-assembly} is the process by which an organized structure
can form spontaneously from simple parts.  It describes the assembly of
diverse natural structures such as crystals, DNA helices, and microtubules.
In nanofabrication, the idea is to co-opt natural self-assembly
processes to build desired structures, such as a sieve for removing viruses
from serum, a drug-delivery device for targeted chemotherapy or brachytherapy,
a magnetic device for medical imaging, a catalyst for enzymatic reactions,
or a biological computer.
Self-assembly of artificial structures has promising applications to
nanofabrication and biological computing.
The general goal is to design and manufacture nanoscale pieces
(e.g., strands of DNA)
that self-assemble uniquely into a desired macroscale object
(e.g., a computer).

Our work is motivated and guided by an ongoing collaboration with the
Sackler School of Graduate Biomedical Sciences
that aims to nanomanufacture sieves, catalysts, and drug-delivery and
medical-imaging devices, using protein self-assembly.
Specifically, the Goldberg Laboratory is currently developing technology
to bioengineer (many copies of) rigid struts of varying lengths,
made of several proteins,
which can join collinearly to each other at compatible ends.
These struts occur naturally as the ``legs'' of the \emph{T4 bacteriophage},
a virus that infects bacteria by injecting DNA.
In contrast to nanoscale self-assembly based on DNA
\cite{Winfree-Liu-Wenzler-Seeman-1998,Mao-Labean-Reif-Seeman-2000,Rothemund-Papadakis-Winfree-2004,Barish-Rothemund-Winfree-2005,Seeman-1998,Shih-Quispe-Joyce-2004,Rothemund-2006},
which is inherently floppy, these nanorod structures are extremely
rigid and should therefore scale up to the manufacture of macroscale
objects.
%
%

The traditional, leading theoretical model for self-assembly is the
two-dimensional \emph{tile assembly model} introduced by Winfree in
his Ph.d. thesis~\cite{Winfree-1998} and
first appearing at STOC 2000 \cite{Rothemund-Winfree-2000}. The
basic building blocks in this model are \emph{Wang tiles}
\cite{Wang-1961}, unrotatable square tiles with a specified glue on
each side, where equal glues have affinity and may stick. Tiles then
self-assemble into supertiles: two (super)tiles nondeterministically
join if the sum of the glue affinities along the attachment is at
least some threshold~$\tau$, called \emph{temperature}.
This basic model has been generalized and extended in many ways
\cite{Adleman-2000-TR,Adleman-Cheng-Goel-Huang-2001,Adleman-Cheng-Goel-Huang-Kempe-Espanes-Rothemund-2002,Soloveichik-Winfree-2004,Aggarwal-Cheng-Goldwasser-Kao-Espanes-Moisset-Schweller-2005,Rothemund-Winfree-2000,Kao-Schweller-2006}.
The model should be practical because Wang tiles can easily simulate
the practical scenario in which tiles are allowed to rotate,
glues come in pairs, and glues have affinity only for their unique mates.
In particular, we can implement such tiles using two unit-length nanorods
joined at right angles at their midpoints to form a plus sign.

Most theoretical research in self-assembly considers the minimum number of
distinct tiles---the \emph{tile complexity}~$t$---required to assemble a shape
uniquely.
In particular, if we allow the desired shape to be scaled by a possibly
very large factor, then in most models the minimum possible tile complexity
(the smallest ``tile program'') is $\Theta(K/\lg K)$ where $K$ is the
Kolmogorov complexity of the shape \cite{Soloveichik-Winfree-2004}.
In practice, the limiting factor is the number of distinct glues---the
\emph{glue complexity}~$g$---as each new glue type requires significant
biochemical research and experiments.  For example, a set of DNA-based glues
requires experiments to test whether a collection of codewords have a
``conflict'' (a pair of noncomplementary base sequences that attach to
each other), while a set of protein-based glues requires finding pairs of
proteins with compatible geometries and amino-acid placements that bind
(and no other pairs of which accidentally bind).
Of course, tile and glue complexities are related: $g \leq t \leq g^4$.

We present the \emph{staged tile assembly model}, a generalization of
the tile assembly model that captures the temporal aspect of the laboratory
experiment, and enables substantially more flexibility in the design and
fabrication of complex shapes using a small tile and glue complexity.
In its simplest form, staged assembly enables the gradual addition of
specific tiles in a sequence of stages.
In addition, any tiles that have not yet attached as part of a supertile
can be washed away and removed
(in practice, using a weight-based filter, for example).
More generally, we can have any number of \emph{bins} (in reality, batches of
liquid solution stored in separate containers), each containing tiles and/or
supertiles that self-assemble as in the standard tile assembly model.
During a stage, we can perform any collection of operations of two types:
(1)~add (arbitrarily many copies of) a new tile to an existing bin; and
(2)~pour one bin into another bin, mixing the contents of the former bin
into the latter bin, and keeping the former bin intact.
In both cases, any pieces that do not assemble into larger structures are
washed away and removed.
These operations let us build intermediate supertiles in isolation
and then combine different supertiles as whole structures.
Now we have two new complexity measures
in addition to tile and glue complexity:
the number of stages---or \emph{stage complexity}~$s$---measures the time
required by the operator of the experiment, while the number of bins---or
\emph{bin complexity}~$b$---measures the space required for the experiment.%
\footnote{Here we view the mixing time required in each stage (and the
  volume of each bin) as a constant,
  mainly because it is difficult to analyze precisely from a thermodynamic
  perspective, as pointed out in \cite{Adleman-2000-TR}.
  In our constructions, we believe that a suitable design of the relative
  concentrations of tiles (a feature not captured by the model) leads to
  reasonable mixing times.}
(When both of these complexities are~$1$, we obtain the regular tile
assembly model.)

\paragraph{Our results.}
We show that staged assembly enables substantially more efficient manufacture
in terms of tile and glue complexity, without sacrificing much in stage
and bin complexity.  All of our results assume the practical constraint of
having only a small constant number of glues and hence a constant number of
tiles.  In contrast, an information-theoretic argument shows that this
assumption would limit the traditional tile assembly model to constructing
shapes of constant Kolmogorov complexity.

For example, we develop a method for self-assembling an
$n \times n$ square for arbitrary $n > 0$, using $16$ glues and
thus $O(1)$ tiles (independent of~$n$), and using only
$O(\log \log n)$ stages, $O(\sqrt{\log n})$ bins, and temperature
$\tau=2$ (Section~\ref{sec:square crazy}).
Alternatively, with the minimum possible temperature $\tau=1$,
we can self-assemble an $n \times n$ square using $9$ glues, $O(1)$ tiles and
bins, and $O(\log n)$ stages (Section~\ref{sec:jigsaw}).
In contrast, the best possible self-assembly of an $n \times n$ square in the
traditional tile assembly model has tile complexity
$\Theta(\log n / \log \log n)$
\cite{Adleman-Cheng-Goel-Huang-2001,Rothemund-Winfree-2000},
or $\Theta(\sqrt{\log n})$ in a rather extreme generalization of
allowable pairwise glue affinities
\cite{Aggarwal-Cheng-Goldwasser-Kao-Espanes-Moisset-Schweller-2005}.

More generally, we show how to self-assemble arbitrary shapes made up
of $n$ unit squares in a variety of precise formulations of the problem.
Our simplest construction builds the shape using $2$ glues, $16$ tiles,
$O(\diameter)$ stages, and $O(1)$ bins,
but it only glues tiles together according to a spanning tree, which
is what we call the \emph{partial connectivity model}
(Section~\ref{sec:spanning tree}). All other constructions have
\emph{full connectivity}: any two adjacent unit squares are built by
tiles with matching glues along their shared edge. In particular, if
we scale an arbitrary hole-free shape larger by a factor of~$2$,
then we can self-assemble with full connectivity using $8$ glues,
$O(1)$ tiles, and $O(n)$ stages and bins (Section~\ref{sec:factor
2}). We also show how to simulate a traditional tile assembly
construction with $t$ tiles by a staged assembly using $3$ glues,
$O(1)$ tiles, $O(\log \log t)$ stages, $O(t)$ bins, and a scale
factor of $O(\log t)$ (Section~\ref{sec:simulation}). If the shape
happens to be monotone in one direction, then we can avoid scaling
and still obtain full connectivity, using $9$ glues, $O(1)$ tiles,
$O(\log n)$ stages, and $O(n)$ bins (Section~\ref{sec:monotone}). We
also discuss an efficient method for the design of binary counters
in the staged assembly framework, an important tool for a large
number of self-assembly systems(Section~\ref{sec:fastCounters}).
This technique offers benefits over non-staged counters in terms of reduced temperature ($\tau=1$) and potentially faster assembly.

\begin{table*}
\centering
\small
\tabcolsep=0.5\tabcolsep
\begin{tabular}{l|c|c|c|c|c|c|c|c}
\multicolumn{1}{c|}{\textbf{$n \times n$ square}} & \textbf{Glues} & \textbf{Tiles} & \textbf{Bins} & \textbf{Stages} & \textbf{$\tau$} & \textbf{Scale} & \textbf{Conn.} & \textbf{Planar}
\\ \hline
Previous work~\cite{Adleman-Cheng-Goel-Huang-2001,Rothemund-Winfree-2000} &
\multicolumn{2}{|c|}{$\Theta(\frac{\log n}{\log\log n})$} & $1$ & $1$ & $2$ & $1$ & full & yes
\\ \hline
Jigsaw technique (\S \ref{sec:jigsaw}) & $9$ & $O(1)$ & $O(1)$ & $O(\log n)$ & $1$ & $1$ & full & yes
\\ \hline
Crazy mixing (\S \ref{sec:square crazy}) & $16$ & $O(1)$ & $B  {}$ &
$O\big(\big\lceil\frac{\log n}{B^2}\big\rceil + \log B\big)$ & $2$ &
$1$ & full & yes
\\ \hline
Crazy mixing, $B=\sqrt{\log n}$ & $16$ & $O(1)$ & $\sqrt{\log n}$ & $O(\log\log n) $ & $2$ & $1$ & full & yes
\\ \hline

\multicolumn{1}{c}{}\\

\multicolumn{1}{c|}{\textbf{General shape with $n$ tiles}} & \textbf{Glues} & \textbf{Tiles} & \textbf{Bins} & \textbf{Stages} & \textbf{$\tau$} & \textbf{Scale} & \textbf{Conn.} & \textbf{Planar}
\\ \hline
Previous work \cite{Soloveichik-Winfree-2004} & \multicolumn{2}{|c|}{$\Theta(K/\log K)$} & $1$ & $1$ & $2$ & unbounded & partial & no
\\ \hline
Arbitrary shape with $n$ tiles (\S \ref{sec:spanning tree}) & $2$ & $16$ & $O(\log n)$ & $O(\diameter)$ & $1$ & $1$ & partial & no
\\ \hline
Hole-free shape with $n$ tiles (\S \ref{sec:factor 2}) & $8$ & $O(1)$ & $O(n)$ & $O(n)$ & $1$ & $2$ & full & no
\\ \hline
Simulation of $1$-stage tiles $T$ (\S \ref{sec:simulation}) & $3$ & $O(1)$ & $O(|T|)$ & $O(\log \log |T|)$ & $1$ & $O(\log |T|)$ & partial & no
\\ \hline
Monotone shapes with $n$ tiles (\S \ref{sec:monotone}) & $9$ & $O(1)$ & $O(n)$ &
$O(\log n)$ & $1$ & $1$ & full & yes
\\ \hline
\end{tabular}
\caption{Summary of the glue, tile, bin, and stage complexities, the
  temperature $\tau$, the scale factor, the connectivity, and the
  planarity of our staged assemblies and the relevant previous work.}
\label{table:summary}
\end{table*}

Table~\ref{table:summary} summarizes our results in more detail, in
particular elaborating on possible trade-offs between the
complexities. The table captures one additional aspect of our
constructions:  Planarity.  Consider two jigsaw puzzle pieces with
complex borders lying on a flat surface.  It may not be possible to
slide the two pieces together while both remain on the table.
Rather, one piece must be lifted off the table and dropped into
position.  Our current model of assembly intuitively permits
supertiles to be placed into position from the third dimension,
despite the fact that it may not be possible to assemble within the
plane.  A \emph{planar} construction guarantees assembly of the
final target shape even if we restrict assembly of supertiles to
remain completely within the plane.  This feature seems desirable,
though it may not be essential in two dimensions because reality
will always have some thickness in the third dimension (2.5D).
However, the planarity constraint (or \emph{spatiality} constraint
in 3D) becomes more crucial in 3D assemblies, \iffull where there is
no fourth dimension to avoid intersection, \fi so this feature gives
an indication of which methods might generalize to 3D; see
Section~\ref{sec:summary}.

\paragraph{Related Work}  There are a handful
of existing works in the field of DNA self-assembly that have
proposed very basic multiple stage assembly procedures.  John Reif
introduced a step-wise assembly model for local parallel
biomolecular computing~\cite{Reif-1999}.  In more recent work
Park~et.~al. have considered a simple hierarchical assembly
technique for the assembly of DNA
lattices~\cite{Park-Pistol-Ahn-Reif-Lebeck-Dwyer-Labean-2006}.
Somei~et.~al. have considered a microfluidic device for stepwise
assembly of DNA tiles~\cite{Somei-Kaneda-Fujii-Murata-2005}.  While
all of these works use some form of stepwise or staged assembly,
they do not study the complexity of staged assembly to the depth
that we do here. Further, none consider the concept of bin
complexity.

\section{The Staged Assembly Model}
\label{sec:model}
In this section, we present basic definitions common to most assembly
models, then we describe the staged assembly model, and finally we define
various metrics to measure the efficiency of a staged assembly system.
\vspace{-.1in}
\paragraph{Tiles and tile systems.} A \emph{(Wang) tile} $t$ is a unit square
defined by the ordered quadruple
$\langle \north(t), \allowbreak \east(t), \allowbreak \south(t), \allowbreak \west(t) \rangle$
of glues on the four edges of the tile.
Each \emph{glue} is taken from a finite alphabet $\Sigma$,
which includes a special ``null'' glue denoted $\nullglue$.
For simplicity of bounds, we do not count the $\nullglue$ glue
in the \emph{glue complexity} $g=|\Sigma|-1$.

A \emph{tile system} is an ordered triple $\langle T, G, \tau \rangle$
consisting of the \emph{tileset} $T$ (a set of distinct tiles),
the \emph{glue function} $G: \Sigma^2 \rightarrow \{0, 1, \dots, \tau\}$,
and the \emph{temperature}~$\tau$ (a positive integer).
It is assumed that $G(x, y) = G(y, x)$ for all $x, y \in \Sigma$ and that
$G(\nullglue,x) = 0$ for all $x \in \Sigma$.
Indeed, in all of our constructions, as in the
original model of Adleman \cite{Adleman-2000-TR},
$G(x,y) = 0$ for all $x \neq y$ (see footnote\footnote{With a typical implementation in DNA, glues actually attach to unique complements rather than to themselves.  However, this depiction of the glue function is standard in the literature and does not affect the power of the model.}),
and each $G(x,x) \in \{1, 2, \dots, \tau\}$.
The \emph{tile complexity} of the system is $|T|$.
\vspace{-.1in}
\paragraph{Configurations.}
Define a \emph{configuration} to be a function
$C: \mathbb{Z}^2 \rightarrow T \cup \{\emptytile\}$,
where $\emptytile$ is a special tile that has the $\nullglue$
glue on each of its four edges. The {\em shape} of a configuration $C$ is
the set of positions $(i, j)$ that do not map to the $\emptytile$ tile.
The shape of a configuration can be disconnected, corresponding to
several distinct supertiles.

\vspace{-.1in}
\paragraph{Adjacency graph and supertiles.}
Define the \emph{adjacency graph} $G_C$ of a configuration $C$ as
follows.  The vertices are coordinates $(i, j)$ such that
$C(i, j) \neq \emptytile$.  There is an edge between two vertices
$(x_1, y_1)$ and $(x_2, y_2)$ if and only if $|x_1-x_2|+ |y_1 - y_2| = 1$.
A \emph{supertile} is a maximal connected subset $G^{\prime}$ of~$G_C$,
i.e., $G^{\prime} \subseteq G_C$ such that, for every connected
subset~$H$, if $G^{\prime} \subseteq H \subseteq G_C$, then $H = G^{\prime}$.
For a supertile $S$, let $|S|$ denote the number of nonempty positions
(tiles) in the supertile.
Throughout this paper, we will informally refer to (lone) tiles as
a special case of supertiles.



If every two adjacent tiles in a supertile share a positive strength
glue type on abutting edges, the supertile is \emph{fully
connected}.


\vspace{-.1in}
\paragraph{Two-handed assembly and bins.}  Informally, in the
two-handed assembly model, any two supertiles may come together
(without rotation or flipping) and attach if their strength of
attachment, from the glue function, meets or exceeds a given
temperature parameter~$\tau$.

Formally, for any two supertiles $X$ and $Y$, the \emph{combination}
set $C^{\tau}_{(X,Y)}$ of $X$ and $Y$ is defined to be the set of
all supertiles obtainable by placing $X$ and $Y$ adjacent
to each other (without overlapping) such that, if we list each newly
coincident edge $e_i$ with edge strength $s_i$, then $\sum s_i \geq \tau$.

We define the assembly process in terms of bins. Intuitively, a bin
consists of an initial collection of supertiles that self-assemble
at temperature $\tau$ to produce a new set of supertiles $P$.
Formally, with respect to a given set of tile-types $T$, a
\emph{bin} is a pair $(S, \tau)$ where $S$ is a set of initial
supertiles whose tile-types are contained in $T$, and $\tau$ is a
temperature parameter.  For a bin $(S, \tau)$, the set of
\emph{produced} supertiles $P'_{(S,\tau)}$ is defined recursively as
follows:  (1)~$S \subseteq P'_{(S,\tau)}$ and (2)~for any $X,Y\in
P'_{(S,\tau)}$, $C^\tau_{(X,Y)} \subseteq P'_{(S,\tau)}$. The set of
\emph{terminally} produced supertiles of a bin $(S,\tau)$ is
$P_{(S,\tau)} = \{ X\in P' \mid Y\in P'$, $C^\tau_{(X,Y)} =
\emptyset \}$.  We say the set of supertiles $P$ is \emph{uniquely}
produced by bin $(S,\tau)$ if each supertile in $P'$ is of finite
size.  Put another way, unique production implies that every
producible supertile can grow into a supertile in $P$.

Intuitively, $P'$ represents the set of all possible supertiles that
can self-assemble from the initial set $S$, whereas $P$ represents
only the set of supertiles that cannot grow any further. In the case
of unique assembly of $P$, the latter thus represents the eventual,
final state of the self-assembly bin. Our goal is therefore to
produce bins that yield desired supertiles in the uniquely produced
set~$P$.

Given a collection of bins, we model the process of mixing bins
together in arbitrarily specified patterns in a sequence of distinct
stages.  In particular, we permit the following actions:  We can
\emph{create} a bin of a single tile type $t \in T$, we can
\emph{merge} multiple bins together into a single bin, and we can
\emph{split} the contents of a given bin into multiple new bins.  In
particular, when splitting the contents of a bin, we assume the
ability to extract only the unique terminally produced set of
supertiles $P$, while filtering out additional partial assemblies in
$P'$.  Intuitively, given enough time for assembly and a large
enough volume of tiles, a bin that uniquely produces $P$ should
consist of almost entirely the terminally produced set $P$.  We
formally model the concept of mixing bins in a sequence of stages
with the \emph{mix graph}.




\vspace{-.1in}
\paragraph{Mix graphs.}
An \emph{$r$-stage $b$-bin mix graph} $M$ consists of $r b+1$
vertices, $m_*$ and $m_{i,j}$ for $1\leq i \leq r$ and $1\leq j \leq
b$, and an arbitrary collection of edges of the form $(m_{r,j},m_*)$
or $(m_{i,j}, m_{i+1,k})$ for some $i, j, k$.
\vspace{-.1in}
\paragraph{Staged assembly systems.}
A \emph{staged assembly system} is a $3$-tuple $\langle M_{r,b}, \{T_{i,j}\},
\{\tau_{i,j}\} \rangle$ where $M_{r,b}$ is an $r$-stage $b$-bin mix
graph, each $T_{i,j}$ is a set of tile types, and each $\tau_{i,j}$
is an integer temperature parameter.  Given a staged assembly system,
for each $1\leq i \leq r$, $1\leq j \leq b$, we define a
corresponding bin $(R_{i,j}, \tau_{i,j})$ where $R_{i,j}$ is
defined as follows:
\begin{enumerate*}
\item $R_{1,j}= T_{1,j}$ (this is a bin in the first stage);
\item For $i\geq 2$,
  $\displaystyle R_{i,j}= \Big(\bigcup_{k:\ (m_{i-1,k},m_{i,j})\in M_{r,b}} P_{(R_{(i-1,k)},\tau_{i-1,k})}\Big) \cup T_{i,j}$.
\item $\displaystyle R_* =\bigcup_{k:\ (m_{r,k},m_{*})\in M_{r,b}} P_{(R_{(r,k)},\tau_{r,k)})}$.
\end{enumerate*}

Thus, the $j$th bin in the $i$th stage takes its initial set of seed
supertiles to be the terminally produced supertiles from a
collection of bins from the previous stage, the exact collection
specified by $M_{r,b}$, in addition to a set of added tile types
$T_{i,j}$. Intuitively, the mix graph specifies how each collection
of bins should be mixed together when transitioning from one stage
to the next.  We define the set of terminally produced supertiles
for a staged assembly system to be $P_{(R_{*}, \tau_*)}$.  In this
paper, we are interested in staged assembly systems for which each
bin yields unique assembly of terminal supertiles.  In this case we
say a staged assembly system uniquely produces the set of supertiles
$P_{(R_{*}, \tau_*)}$.

Throughout this paper, we assume that, for all $i,j$,
$\tau_{i,j} = \tau$ for some fixed global temperature~$\tau$,
and we denote a staged assembly system as
$\langle M_{r,b}, \{T_{i,j}\}, \tau \rangle$.

\iffull
\begin{figure}[h]
\centering
\includegraphics[scale=1.5]{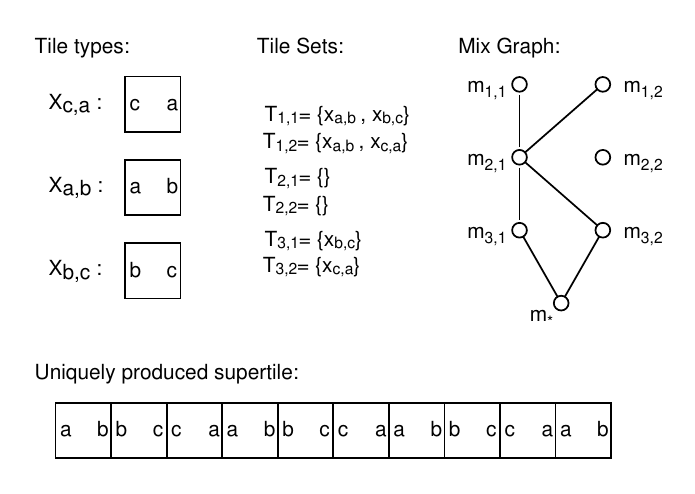}
\caption{A sample staged assembly system that uniquely assembles a
  $1\times 10$ line.  The temperature is $\tau=1$,
  and each glue $a,b,c$ has strength~$1$.
  The tile, stage, and bin complexities are~$3$, $3$, and $2$, respectively.}
\label{figure:StagedExample}
\end{figure}
\fi

\iffull
\paragraph{Metrics.}
\label{sec:metrics}
We are interested in designing efficient staged assembly systems that
terminally produce a unique target shape.  We use the following natural
metrics to measure the efficiency of the staged tile system:
\begin{description*}
\item[Tile complexity:]  $|\bigcup T_{i,j}|$.  This
represents the number of distinct tile types that the assembly
system requires.  In this paper we emphasize $O(1)$ tile complexity
systems in contrast to previous work.
\item[Bin complexity:]  The number $b$ of vertices in each partition of the
mix graph.  Intuitively this measures the number of distinct containers that
would be required to carry out the specified staged assembly
procedure.
\item[Stage complexity:] The number $r$ of sequential stages of mixing
that occur.  This metric
measures the number of stages in which collections of bins must
be brought to their terminal assemblies and mixed together into a
new array of bins. It represents operator time.
\item[Temperature:]  The value $\tau$.
In practice, it is difficult to implement systems with
accurate temperature sensitivity.
In this paper we focus on $\tau \in \{1,2,3\}$.
\end{description*}
We also consider the following features
to measure the quality of the shape produced:
\begin{description*}
\item[Planarity:]
  In a \emph{planar} construction, supertiles have obstacle-free paths
  to reach their mates.
\item[Connectivity:]
  In a \emph{fully connected} supertile, every two adjacent tiles have
  the same positive-strength glue along their common edge.
  Otherwise the supertile is \emph{partially connected}.
\item[Scale factor:]
  In some cases, we allow the produced shape to a uniform scaling of
  the desired shape by some small positive integer,
  called the \emph{scale factor}.
\end{description*}

\begin{figure}[h]
\centering
\includegraphics[scale=1.35]{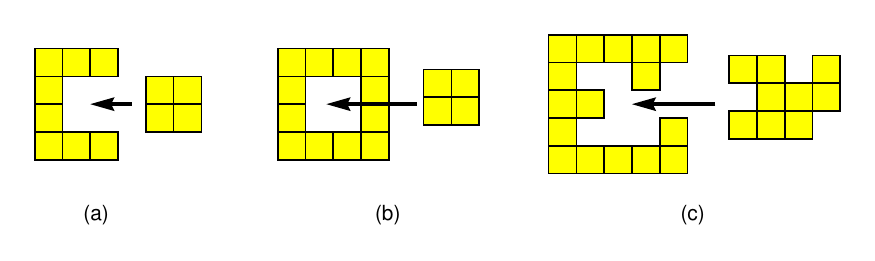}
\caption{All three assemblies are permitted under the basic model.
However, only assembly (a) is permitted under the planarity
constraint.} \label{figure:Planarity}
\end{figure}

\fi

\section{Assembly of $1\times n$ Lines}
\label{sec:line}

As a warmup, we develop a staged assembly for the $1\times n$ rectangle
(``line'') using only three glues and $O(\log n)$ stages.
\iffull

\fi
The assembly uses a divide-and-conquer approach
to split the shape into a constant number of recursive pieces.
Before we turn to the simple divide-and-conquer required here,
we describe the general case, which will be useful later.
This approach requires the pieces to be combinable in a
unique way, forcing the creation of the desired shape.
We consider the \emph{decomposition tree} formed by the recursion,
where sibling nodes should uniquely assemble to their parent.
The staging proceeds bottom-up in this tree.
The height of this tree corresponds to the stage complexity, and the maximum
number of distinct nodes at any level corresponds to the bin complexity.
The idea is to assign glues to the pieces in the decomposition tree
to guarantee unique assemblage while using few glues.

\iffull
We now turn to constructing $1 \times 2^k$ lines:
\fi

\begin{theorem} \label{powers of two}
  There is a planar temperature-$1$ staged assembly system that uniquely produces a
  (fully connected) $1 \times 2^k$ line using $3$ glues,
  $6$ tiles, $6$ bins, and $O(k)$ stages.
\end{theorem}

\begin{proof}
  The decomposition tree simply splits a $1 \times 2^k$ line into
  two $1 \times 2^{k-1}$ lines.
  All tiles have the $\nullglue$ glue on their top and bottom edges.
  If the $1 \times 2^k$ line has glue $a$ on its left edge,
  and glue $b$ on its right edge, then the left and right $1 \times 2^{k-1}$
  inherit these glues on their left and right edges, respectively.
  We label the remaining two inner edges---the right edge of the left piece
  and the left edge of the right piece---with a third glue~$c$,
  distinct from $a$ and~$b$.
  Because $a \neq b$, the left and right piece uniquely attach at the
  inner edges with common glue~$c$.
  This recursion also maintains the invariant that $a \neq b$,
  so three glues suffice overall.
  Thus there are only ${3 \choose 2} = 6$ possible $1 \times 2^k$ lines of
  interest, and we only need to store these six at any time, using six bins.
  At the base case of $k=0$, we just create the six possible single tiles. 
  The number of stages beyond that creation is exactly~$k$.
\end{proof}

\iffull
\begin{figure}[h]
\centering
\includegraphics[scale=.50]{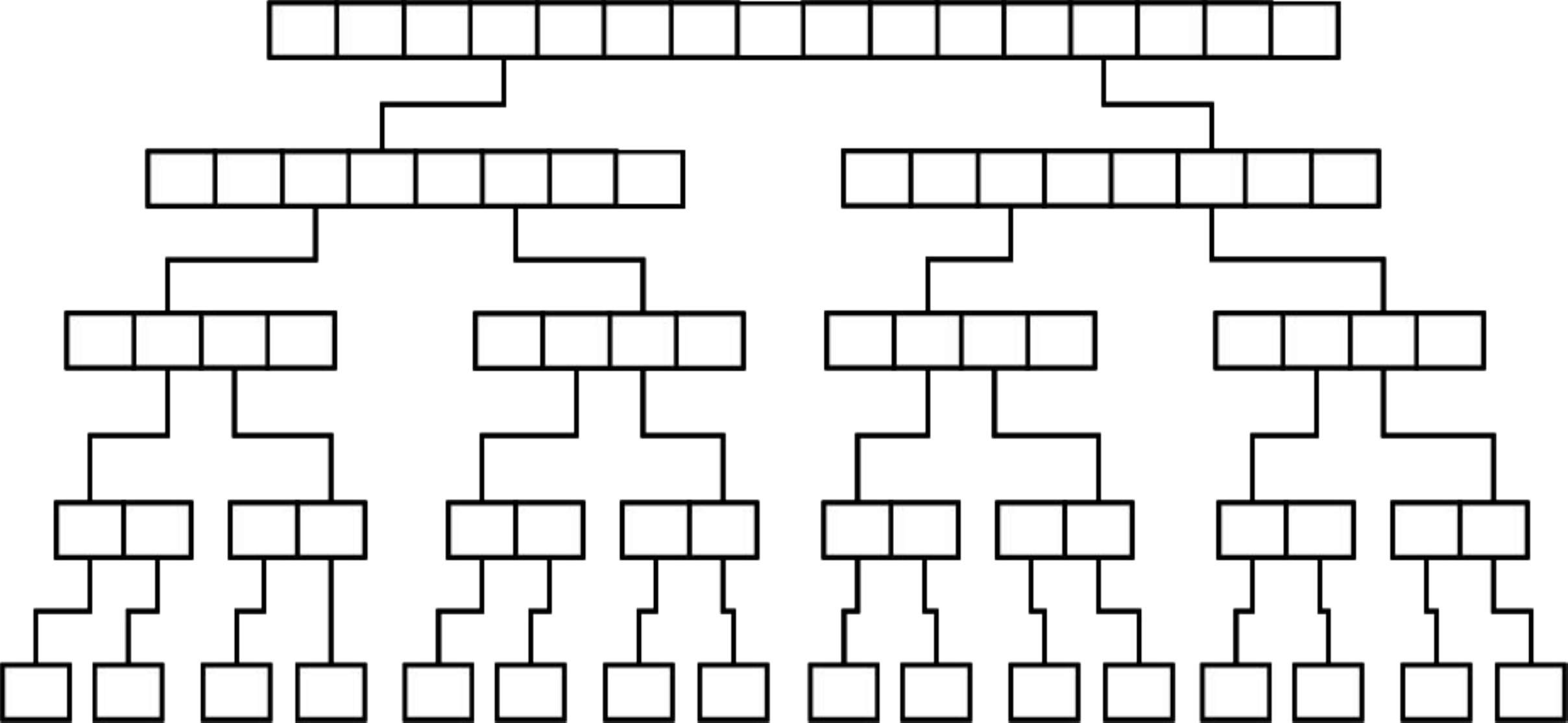}
\caption{Decomposition tree for $1 \times 16$ line.}
\label{figure:DecompositionTree1xN}
\end{figure}
\fi

\begin{corollary}
  There is a planar temperature-$1$ staged assembly system that uniquely produces a
  (fully connected) $1 \times n$ line using $3$ glues,
  $6$ tiles, $7$ bins, and $O(\log n)$ stages. 
\end{corollary}

\begin{proof}
  We augment the construction of Theorem~\ref{powers of two}
  applied to $k = \lfloor \log n \rfloor$.
  When we build the $1 \times 2^i$ lines for some~$i$,
  if the binary representation of $n$ has a $1$ bit in the $i$th position,
  then we add that line to a new output bin.
  Thus, in the output bin, we accumulate powers of $2$ that sum to~$n$.
  As in the proof of Theorem~\ref{powers of two},
  three glues suffice to guarantee unique assemblage in the output bin.
  The number of stages remains $O(\log n)$.
\end{proof}

\xxx{work out the constants if we assemble an $n \times n$ as if it were
a comb, and add theorem and entry to table}

\section{Assembly of $n \times n$ Squares}
\label{sec:square}

Figure~\ref{figure:ShiftingProblem} illustrates the challenge with generalizing
the decomposition-tree technique from $1 \times n$ lines to $n \times n$
squares.
Namely, the na\"{\i}ve decomposition of a square into two $n \times n/2$
rectangles cannot lead to a unique assembly using $O(1)$ glues with
temperature~$1$ and full connectivity: by the pigeon-hole principle,
some glue must be used more than once along the shared side of length~$n$,
and the lower part of the left piece may glue to the higher part
of the right piece.  Even though this incorrect alignment may make two unequal
glues adjacent, in the temperature-$1$ model, a single matching pair of glues
is enough for a possible assembly.

\iffull
\begin{figure}[h]
\centering
\subfigure[\label{figure:ShiftingProblem}]
  {\vbox{\hbox{\includegraphics[scale=0.5]{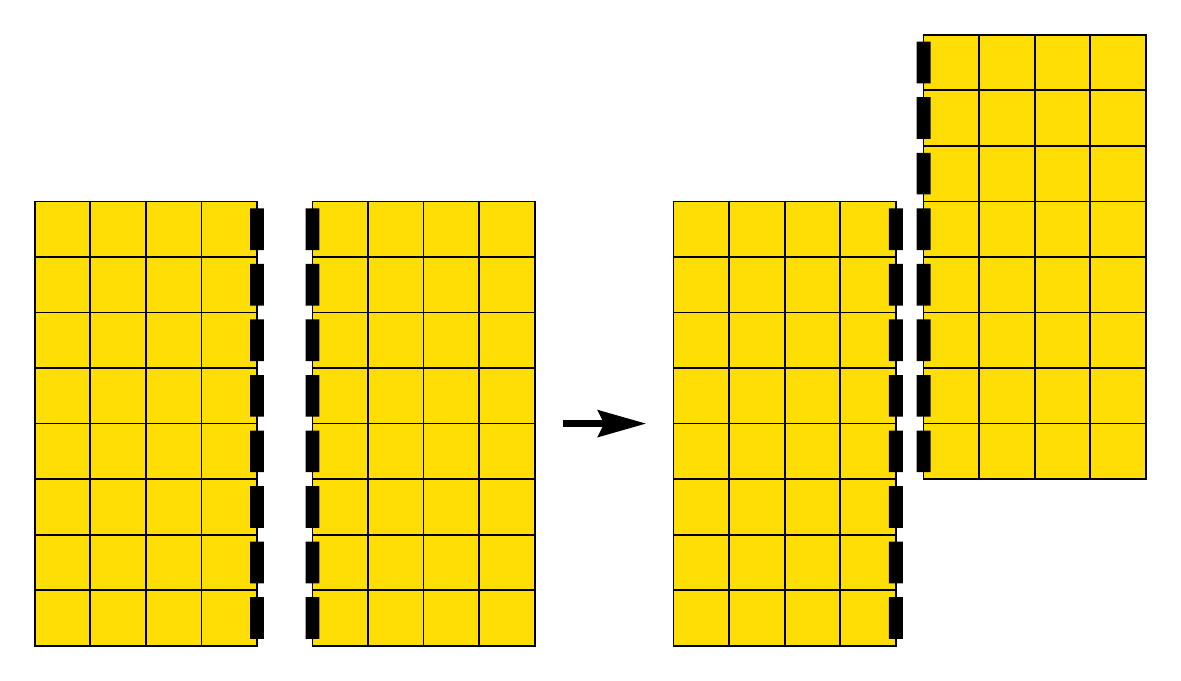}}\vspace{-0.15in}}}\hfil\hfil
\subfigure[\label{figure:Jigsaw}]
  {\vbox{\hbox{\includegraphics[scale=0.5]{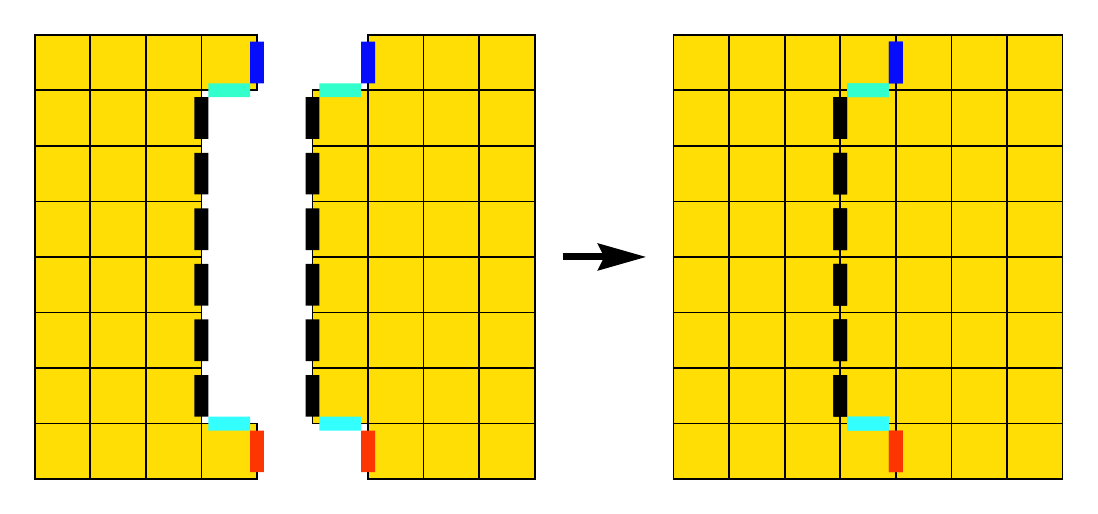}}\vspace{0.2in}}}
\caption{(a)~The shifting problem encountered when combining rectangle
  supertiles.
  (b)~The jigsaw solution: two supertiles that combine uniquely
  into a fully connected square supertile.}
\end{figure}
\fi

\subsection{Jigsaw Technique}
\label{sec:jigsaw}

To overcome this shifting problem, we introduce the \emph{jigsaw technique},
a powerful tool used throughout this paper.
This technique ensures that the two supertiles glue together uniquely
based on geometry instead of glues.
Figure~\ref{figure:Jigsaw} shows how to cut a square supertile into
two supertiles with three different glues
that force unique combination while preserving full connectivity.

\begin{theorem}
  There is a planar temperature-$1$ staged assembly of a fully connected
  $n \times n$ square using $9$ glues, $O(1)$ tiles, $O(1)$ bins, and
  $O(\log n)$ stages.
\end{theorem}

\begin{proof}
  We build a decomposition tree by first decomposing the $n \times n$ square
  by vertical cuts, until we obtain tall, thin supertiles; then we similarly
  decompose these tall, thin supertiles by horizontal cuts, until we obtain
  constant-size supertiles.
  Table~\ref{table:DecompositionAlgorithm} describes the general algorithm.
  Figure~\ref{figure:DecompositionTree} shows the decomposition tree
  for an $8 \times 8$ square.
  The height of the decomposition tree, and hence the stage complexity,
  is $O(\log n)$.

\begin{table}[h]
\centering
\small
\fbox{\begin{minipage}{5.8in}
\smallskip
\textbf{Algorithm} DecomposeVertically (supertile $S$): \\
\hbox to \hsize{\hfil --- Here $S$ is a supertile with $n$ rows and $m$ columns; $S$ is not necessarily a rectangle.}
\begin{enumerate*}
\parindent=1.5em
\item \textbf{Stop vertical partitioning when width is small enough:}

  If $m \leq 3$, DecomposeHorizontally$(S)$ and return.

\item \textbf{Find the column along which the supertile is to be partitioned:}

  Let $i := \lfloor (m + 1)/2 \rfloor$.

  Divide supertile $S$ along the $i$th column
  into a left supertile $S_1$ and right supertile $S_2$ such that

  tiles at position $(1, i)$ and $(n, i)$
  belong to $S_1$ and the rest of the $i$th column belongs to~$S_2$.

\item \textbf{Now decompose recursively:}

  DecomposeVertically $(S_1)$

  DecomposeVertically $(S_2)$
\end{enumerate*}
\smallskip
\end{minipage}}
\caption{Algorithm for vertical decomposition.
         (Horizontal decomposition is symmetric.)}
\label{table:DecompositionAlgorithm}
\end{table}

\begin{figure}[h]
  \centering
  \includegraphics[scale=0.35]{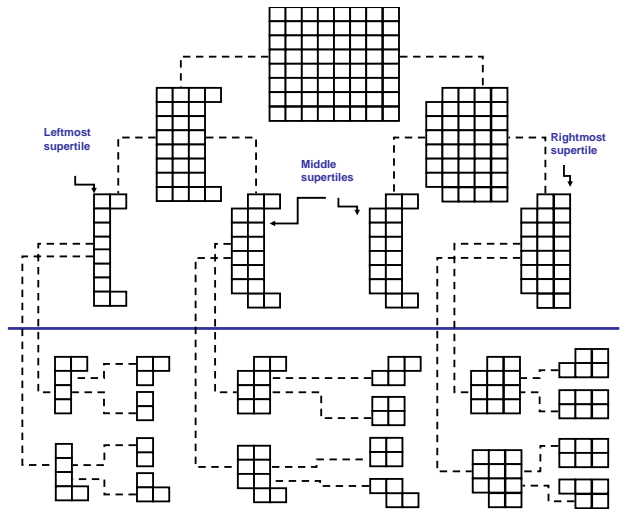}
  \caption{Decomposition tree for $8 \times 8$ square in the jigsaw technique.}
  \label{figure:DecompositionTree}
\end{figure}

We assign glue types to the boundaries of the supertiles to guarantee
unique assemblage based on the jigsaw technique.
The assignment algorithm is similar to the $1 \times n$ line,
but we use three glues for the boundary of each supertile instead of one,
for a total of nine glues instead of three.
\iffull
Figure~\ref{figure:GlueAssignment} shows the glue assignment
during the first two vertical decompositions of the $8 \times 8$ square.
\fi

\iffull
\begin{figure}[h]
  \centering
  \includegraphics[scale=.50]{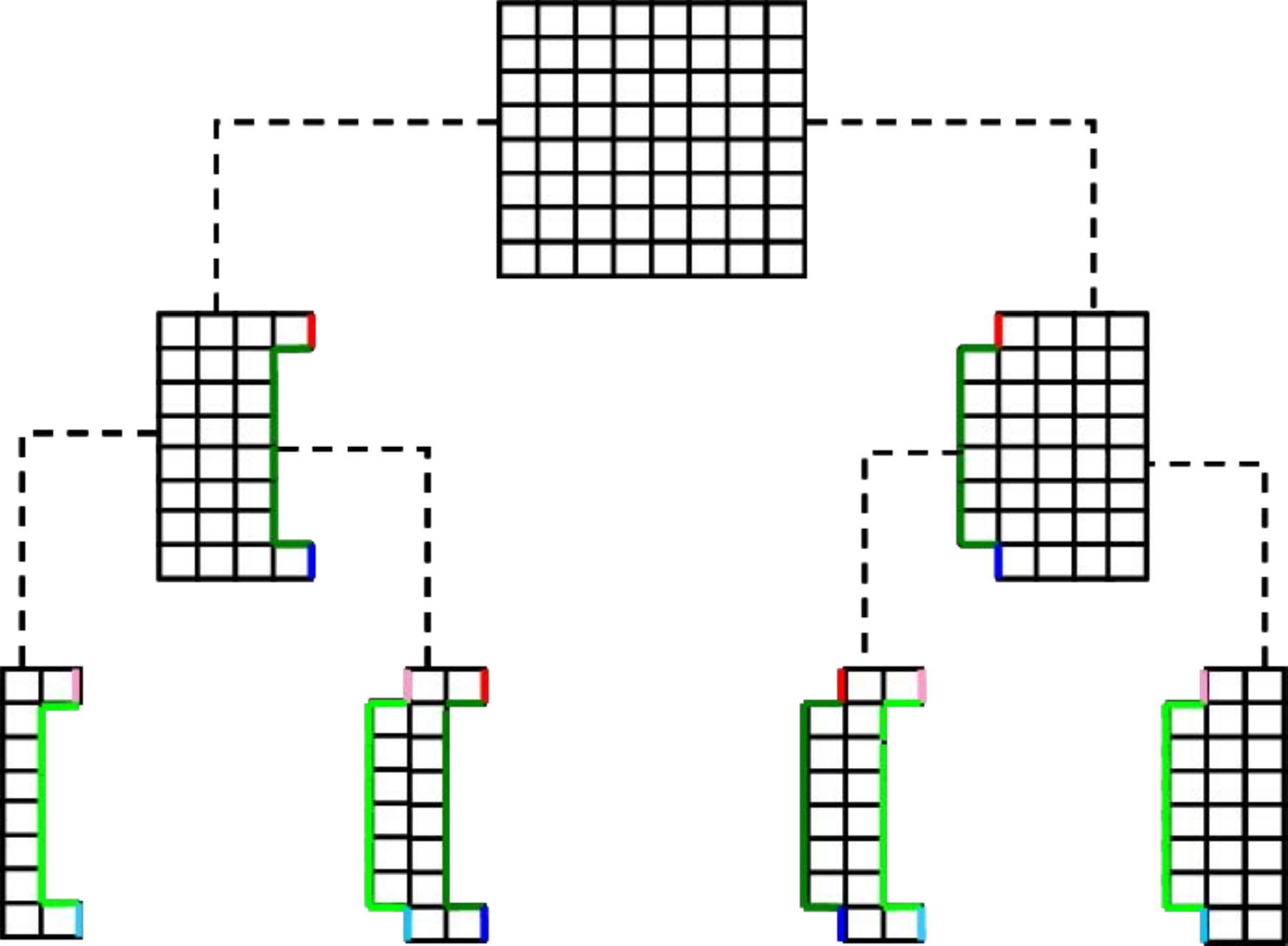}
  \caption{Assigning glues in the first two vertical decompositions of the
           jigsaw technique.}
  \label{figure:GlueAssignment}
\end{figure}
\fi

It remains to show that the bin complexity is $O(1)$.
We start by considering the vertical decomposition.
At each level of the decomposition tree, there are three types of intermediate
products: leftmost supertile, rightmost supertile and middle supertiles.  The
leftmost and rightmost supertiles are always in different bins. The important
thing to observe is that the middle supertiles always have the same shape,
though it is possible to have two different sizes---the number of columns can
differ by one.  In one of these sizes, the number of columns is even and,
in the other, the number is odd.  Thus we need separate bins for the even-
and odd-columned middle supertiles.  For each of the even- or odd-columned
supertiles, each of left and right boundaries of the supertile can have three
choices for the glue types.  Therefore, there is a constant number of different
types of middle supertiles at each level of the decomposition tree.  Thus, for
vertical decomposition, we need $O(1)$ bins.  Each of the supertiles at the end
of vertical decomposition undergoes horizontal decomposition.
A similar argument applies to the horizontal decomposition as well.
Therefore, the number of bins required is $O(1)$.
\end{proof}

\subsection{Crazy Mixing}
\label{sec:square crazy}

For each stage of a mix graph on $B$ bins, there are up to $\Theta(B^2)$ edges that can be included in the mix graph.  By picking which of these edges are included in each stage, $\Theta(B^2)$ bits of information can be encoded into the mix graph per stage.  The large amount of information
that can be encoded in the mixing pattern of a stage permits a very
efficient trade-off between bin complexity and stage complexity.  In
this section, we consider the complexity of this trade-off in the
context of building $n\times n$ squares.

It is possible to view a tile system as a compressed encoding of the
shape it assembles.  Thus, information theoretic lower bounds for
the descriptional or Kolmogorov complexity of the shape assembled
can be applied to aspects of the tile system.  From this we obtain
the following lower bound:


\begin{theorem}\label{thm:crazylower}
  Any staged assembly system with a fixed temperature and bin complexity
  $B$ that uniquely assembles an $n\times n$ square with $O(1)$ tile
  complexity must have stage complexity $\Omega(\frac{\log n}{B^2})$
  for almost all~$n$.
\end{theorem}
\begin{proof}
The Kolmogorov complexity of an integer $n$ with respect to a
universal Turing machine $U$ is $K_U (n) = $ min$|p|$ s.t $U(p) =
b_n$ where $b_n$ is the binary representation of $n$.  A
straightforward application of the pigeonhole principle yields that
$K_U (n) \geq \lceil \log n \rceil - \Delta$ for at least $1-
(\frac{1}{2})^{\Delta}$ of all $n$ (see~\cite{Li:1997:IKC} for
results on Kolmogorov complexity).  Thus, for any $\epsilon > 0$,
$K_U (n) \geq (1- \epsilon)\log n = \Omega (\log n)$ for almost all
$n$.

There exists a fixed size Turing machine that takes as input a
staged assembly system and outputs the maximum length of the
uniquely assembled shape of the system, if there is one.  Such a
machine that takes as input a system $S=\langle M_{r,b},
\{T_{i,j}\}, \tau \rangle$ that uniquely assembles an $n\times n$
square will output the integer $n$, and therefore must have size at
least $K_U (n)$.  Therefore, an encoding of $S$ into bits must have
size at least $\Omega(\log n)$ for almost all $n$. But, for a
constant bounded $\tau$ and $|T| = O(1)$, we can encode
$\{T_{i,j}\}$ and $\tau$ in $O(rb)$ bits and $M_{r,b}$ in $O(rb^2)$
bits for a total $O(rb^2)$ length encoding. Thus, for some constants
$c_1$ and $c_2$ we know that for almost all $n$, $c_1 rb^2 \geq c_2\log n$, which yields, $r \geq \frac{c_2 \log n}{c_1 b^2}$.

\end{proof}

Our upper bound achieves a stage complexity that is within a $O(\log
B)$ additive factor of this lower bound:

\begin{theorem}\label{thm:crazyupper}
  For any $n$ and $B$, there is a temperature-$2$ fully connected
  staged assembly of an $n\times n$ square using $16$ glues,
  $O(1)$ tiles, $B$ bins, and $O(\frac{\log n}{B^2} + \log B)$ stages.
\end{theorem}
\begin{proof}
\begin{figure}[t]
\centering
\includegraphics[scale=.4]{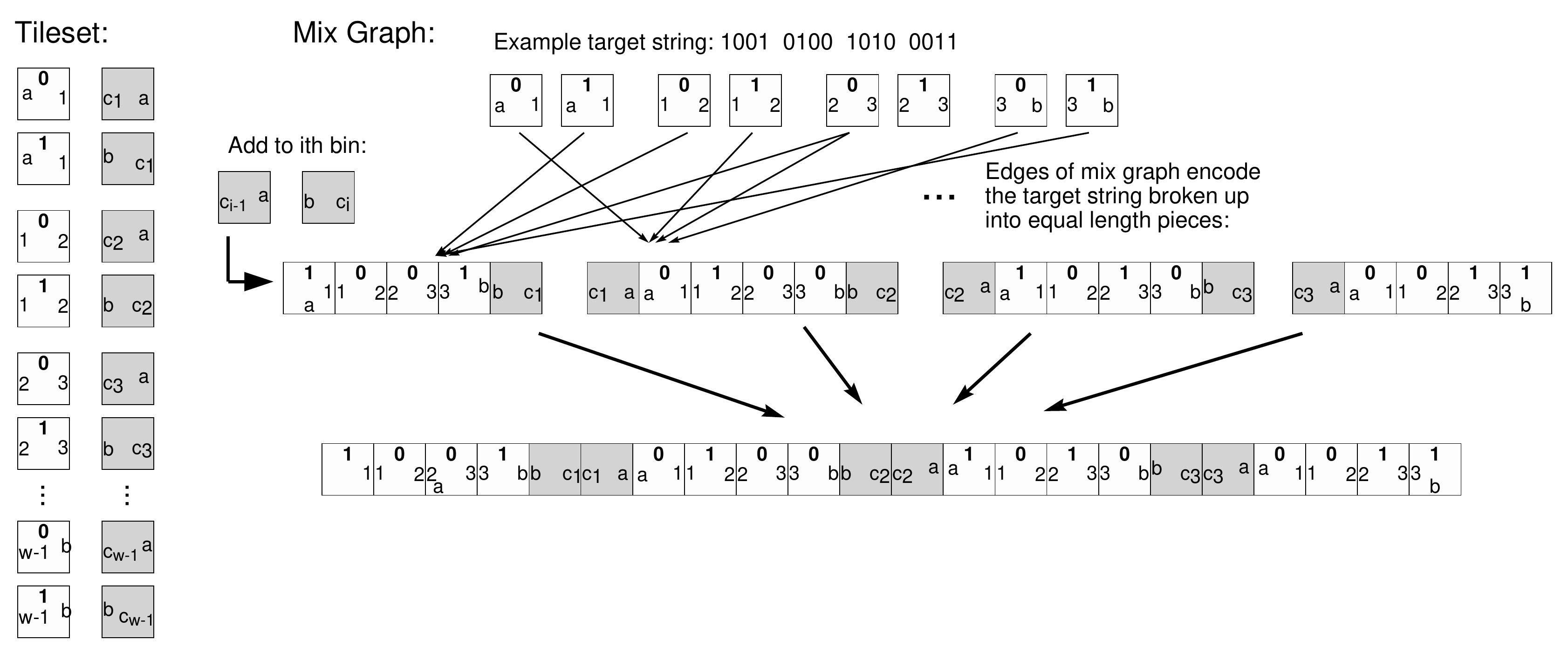}
\caption{This tileset and mix graph depict a tile system with $2w$
tiles and $w$ bins that will assemble an arbitrarily specified
length $w^2$ binary string.  } \label{figure:CrazySquare}
\end{figure}

\begin{figure}[t]
\centering
\includegraphics[scale=.85]{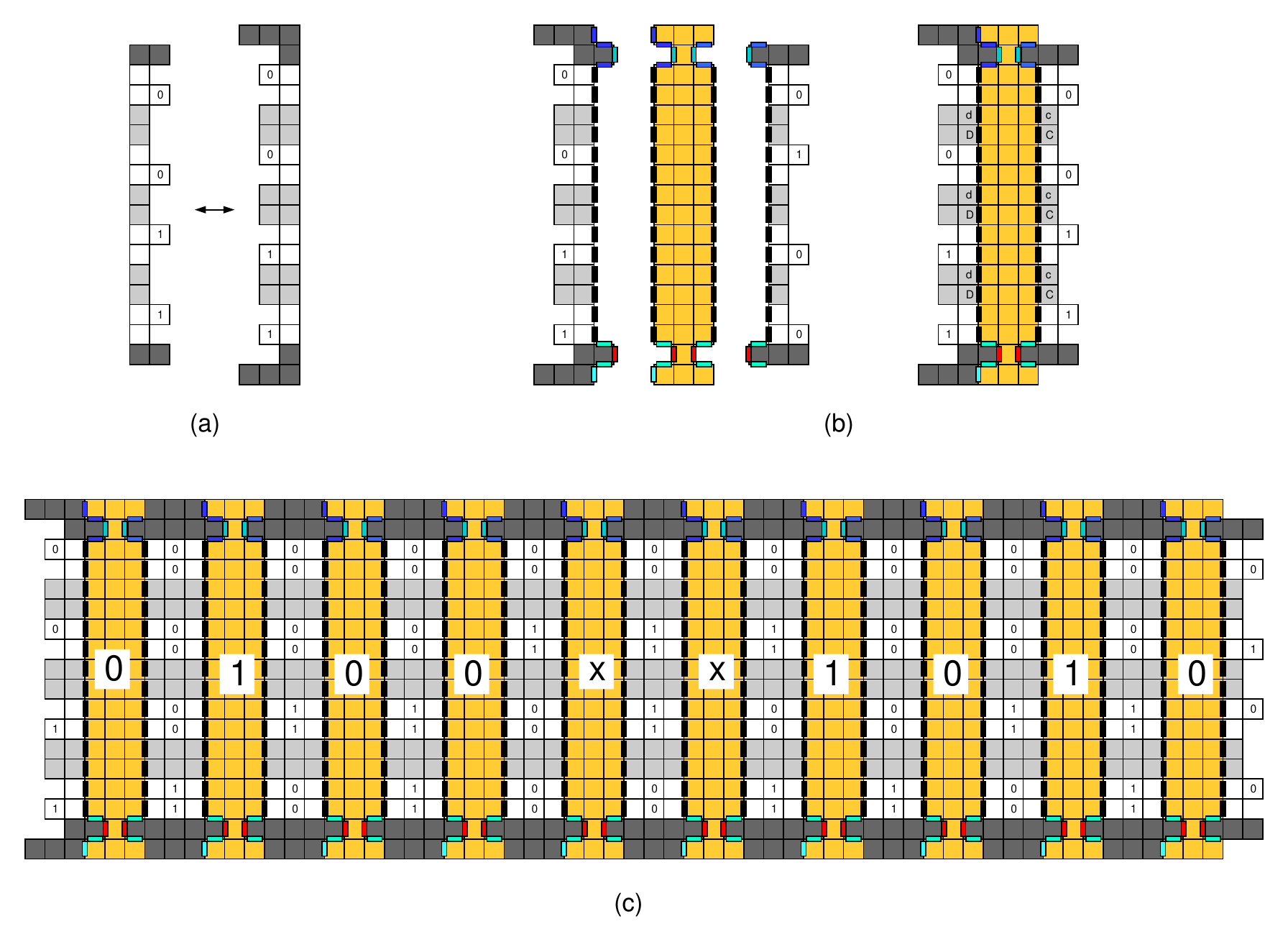}
\caption{ (a) As described in Section~\ref{sec:simulation}, a
collection of supertiles, each in its own bin, can be created such
that each supertile encodes a binary string with a sequence of
pockets and tabs specified by the binary pattern encoded.  These
supertiles act as large glues, or macro glues, by combining glue
type and geometry to bond only to their exact complement supertile.
(b) By combining macro glues onto the surfaces of supertiles, macrotiles
can be created whose assembly pattern is the same as that for a
corresponding set of singleton tiles.  (c)  The tile set of from
Figure~\ref{figure:CrazySquare} can be simulated with macro tiles to
create a binary string using only $O(1)$ tile complexity. The
example assembly shown here corresponds to the middle section of the
example assembly from Figure~\ref{figure:CrazySquare}.}
\label{figure:CrazyDetails}
\end{figure}

Within a $O(1)$ additive factor of tile complexity,
\cite{Rothemund-Winfree-2000} have reduced the problem of assembling
an $n\times n$ square at temperature 2 to the assembly of a
length-$\log n$ binary string that uniquely identifies~$n$.  A
straightforward adaptation of the analysis shows that this result
also works in the two-handed assembly model used in this paper.
Therefore, we focus simply on building an arbitrary $x$-bit input
binary string to prove the theorem.

We first show how to build a length-$O(B^2)$ bit string using a
temperature-$1$ system that makes use of $B$ bins, $O(B)$ distinct
tiles, and $O(1)$ stages. We then 
apply a technique similar to that of Theorem~\ref{thm:simulation} to
convert this system into a $O(1)$ tile complexity system with an
addition of $O(\log B)$ stages. Finally, to get all $x$ bits we can
repeat this process $\lceil\frac{x}{B^2}\rceil$ times for a total of
$O(\frac{x}{B^2} + \log B)$ stage complexity.

For some arbitrary integer $w$, consider the size $2w$ tileset and
corresponding $3$-stage mix graph given in
Figure~\ref{figure:CrazySquare}.  For each of $w$ bit positions,
there is a corresponding pair of white tiles, one representing the
binary value $0$, the other representing $1$.  By placing exactly
one white tile from each pair into a single bin, a length $w$ bit
string is specified.  In the transition from stage 1 to stage 2,
such a length $w$ string is built for each of $w$ bins, yielding $w$
length $w$ bit strings.  These strings can then be concatenated in
the transition from stage 2 to 3 to yield a length $w^2$ binary
string.

For $w= B/2$, this yields a system with $O(B)$ bins, $O(B)$ tile
complexity, and $O(1)$ stage complexity that assembles a
length-$B^2$ target string.  To reduce the tile complexity to
$O(1)$, we apply a technique similar to that of
Theorem~\ref{thm:simulation}.  In particular, we use $O(B)$ bins and
$O(log B)$ stages to create a size $B$ alphabet of \emph{macro}
glues as shown in Figure~\ref{figure:CrazyDetails}. Each macro glue
is a supertile that consists of a string of tiles representing bits
of a binary string.  Further, with the same bin and stage complexity
we create a parallel set of \emph{complement} macro glues as shown
in Figure~\ref{figure:CrazyDetails} (a).  Note that when combined
into the same bin, two macro glues will only attach to one another
if they are exact complements (have the same binary encoding).  By
design, the tooth like geometry of the macro glues provides that
even a single bit difference between two macro glues excludes even a
single bond from attaching.

Given this set of macro glues, we now conceptually index the set of
distinct glues from the size $O(B)$ tileset of
Figure~\ref{figure:CrazySquare} and assign each glue a corresponding
macro glue whose binary string matches the index of the glue. Next,
we attach macro glues to the long thin supertiles shown in
Figure~\ref{figure:CrazyDetails} which can be created using a
slightly modified version of the line algorithm from
Section~\ref{sec:line}. In particular, for each element of the
tileset from Figure~\ref{figure:CrazySquare}, we attach macro glues
corresponding to the east and west glues of the singleton tile.
Further, we can assign a glue representing '0', '1', or 'nothing' on
the north surface of each macro tile, according to which glue the
corresponding singleton tile displays on its north side.

Once we have built this set of \emph{macro tiles}, we mix them
according to the same mixing algorithm for the size $O(B)$ tile set,
but instead replace each singleton tile with its corresponding macro
tile.  By design, the macro glues attach exactly as the basic glues
they are built to emulate.  The result is thus a length-$w^2$ binary
string encoded on the north surface of the assembled macro tiles.


Finally, to get a length-$x$ string, we can repeat this process
$\lceil\frac{x}{B^2}\rceil$ times for a total of $O(\frac{x}{B^2}
+\log B)$ stage complexity.
Given this string, the technique of~\cite{Rothemund-Winfree-2000} is
easily adapted to take into account the $\log B$ vertical
magnification factor we introduce by utilizing the macro glue
construction.  
Further, while the technique of ~\cite{Rothemund-Winfree-2000} is
temperature 2 rather than 1, this is not a problem as we can simply
double the strength of each glue in our construction to make it a
temperature 2 system.  Details of applying the square building set
from~\cite{Rothemund-Winfree-2000} are straightforward and
applications of the technique to similar problems are considered
in~\cite{Aggarwal-Cheng-Goldwasser-Kao-Espanes-Moisset-Schweller-2005,Kao-Schweller-2006}.

Finally, we observe that the construction used here can be designed
to achieve full connectivity and is planar.  Further, the
construction of~\cite{Rothemund-Winfree-2000} maintains this full
connectivity and planarity, yielding the result.

\end{proof}


\xxx{need to add $16$-glue argument}




We conjecture that this stage complexity bound can be achieved by a
temperature-$1$ assembly by judicious use of the jigsaw technique.
\xxx{do this}

\section{Assembly of General Shapes}
\label{sec:general_shapes}

In this section, we describe a variety of techniques for manufacturing
arbitrary shapes using staged assembly with $O(1)$ glues and tiles.


\subsection{Spanning-Tree Technique}
\label{sec:spanning tree}

The \emph{spanning-tree technique} is a general tool for making an
arbitrary shape with the connectivity of a tree.  We start with
a sequential version of the assembly:

\begin{theorem} \label{SpanningTreeTheorem}
  Any shape $S$ with $n$ tiles has a partially connected temperature-$1$
  staged assembly using $2$ glues, at most $16$ tiles, $O(\log n)$ bins,
  and $O(\diameter(S))$ stages.
\end{theorem}

\iffull
\begin{figure}[h]
\centering
\includegraphics[scale=.70]{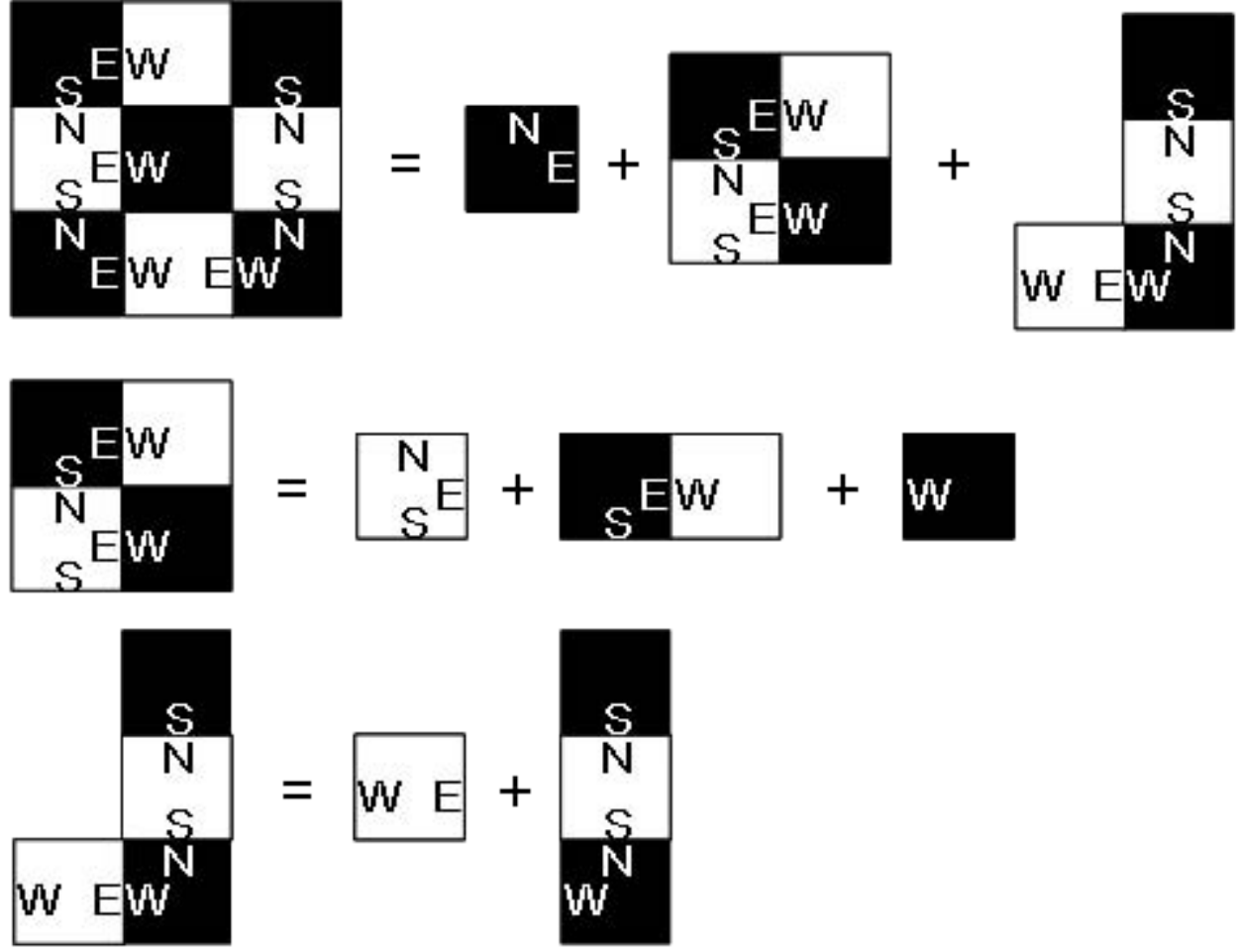}
\xxx{This figure needs fixing to use just two colors, no N/S/E/W.}
\caption{Spanning-tree method for assembling a $3 \times 3$ square.}
\label{figure:SpanningTree}
\end{figure}
\fi

\begin{proof}
  Take a breadth-first spanning tree of the adjacency graph of the shape~$S$.
  The depth of this tree is $O(\diameter(S))$.
  Root the tree at an arbitrary leaf.
  Thus, each vertex in the tree has at most three children.
  Color the vertices with two colors, black and white,
  alternating per level.
  For each edge between a white parent and a black child,
  we assign a white glue to the corresponding tiles' shared edge.
  For each edge between a black parent and a white child,
  we assign a black glue to the corresponding tiles' shared edge.
  All other tile edges receive the $\nullglue$ glue.
  Now a tile has at most three edges of its color connecting to its children,
  and at most one edge of the opposite color connecting to its parent.

  To obtain the sequential assembly, we perform a particular postorder
  traversal of the tree: at node~$v$, visit its child subtrees in
  decreasing order of size.  To combine at node~$v$, we mix the recursively
  computed bins for the child subtrees together with the tile corresponding
  to node~$v$.
  The bichromatic labeling ensures unique assemblage.
  The number of intermediate products we need to store is $O(\log n)$,
  because when we recurse into a second child, its subtree must have size
  at most $2/3$ of the parent's subtree.
  \xxx{probably should elaborate on that argument slightly}
\end{proof}

\iffull
Figure~\ref{figure:SpanningTree} illustrates
spanning tree method for assembling $3 \times 3$ square.
In general, this construction is nonplanar:
the trees may fit together like a key in a keyhole.
\fi

\iffull
\xxx{The following seems false?!}

The stage complexity of the spanning-tree technique can be reduced by
parallelization, at the cost of more bins:

\begin{theorem}
  Any shape $S$ with $n$ tiles has a partially connected temperature-$1$
  staged assembly using $O(1)$ tiles, $O(\log n)$ stages, and
  $O(n/\log n)$ bins.
\end{theorem}

\begin{proof}
  As before, we consider a two-colored breadth-first spanning tree.
  To build an $n$-tile tree,
  split this tree into two trees of at most $2n/3$ nodes.
  Recursively build these two trees, and then mix the two resulting bins of
  supertiles together.
  If we continue this recursion down to individual nodes (tiles),
  we get $n$ such trees and the stage complexity reduces to $O(\log n)$,
  but the bin complexity is now $O(n)$.
  We can do better if we recurse until we get
  $n/{\log n}$ trees of size $\log n$ each.
  By Theorem \ref{SpanningTreeTheorem},
  each of these trees can be built using $O(\log n)$ stages and $O(1)$ bins.
  Thus we need $n/\log n$ bins in total.
  These trees can be combined using another $O(\log n)$ stages to get the
  $n$-tile tree.
\end{proof}
\fi

\subsection{Scale Factor $2$}
\label{sec:factor 2}

Although the spanning-tree technique is general, it probably manufactures
structurally unsound assemblies.
Next we show how to obtain full connectivity of general shapes,
while still using only a constant number of glues and tiles.

\begin{theorem} \label{scale2}
  Any simply connected shape has a staged assembly
  using a scale factor of~$2$, $8$~glues, $O(1)$ tiles,
  $O(n)$ stages, and $O(n)$ bins.
  \xxx{Possible to get $O(\sqrt n)$ stages or $O(\sqrt n)$ bins?}
  The construction maintains full connectivity.
\end{theorem}

\begin{proof}
  Slice the target shape with horizontal lines to divide the shape
  into $1 \times k$ strips for various values of~$k$,
  which scale to $2 \times 2k$ strips (uniform factor-2 scaling of the target shape).
  These strips can be adjacent along horizontal edges but not along vertical edges. 
  Define the \emph{strip graph} to have a vertex for each strip
  and an edge between two strips that are adjacent along a horizontal edge.
  Because the shape is simply connected (hole-free), the strip graph is a tree.
  Root this tree at an arbitrary strip, defining a parent relation.

  A recursive algorithm builds the subtree of the strip graph rooted at
  an arbitrary strip~$s$.
  As shown in Figure~\ref{figure:scale2}(a), the strip $s$
  may attach to the rest of the shape at zero or more places
  on its top or bottom edge.
  One of these connections corresponds to the parent of~$s$
  (unless $s$ is the overall root).
  As shown in Figure~\ref{figure:scale2}(b), our goal is to form
  each of these attachments using a jigsaw tab/pocket combination,
  where bottom edges have tabs and top edges have pockets, extending
  from the rightmost square up to but not including the leftmost square. Factor-2 scaling ensures that it is always possible to create these tabs and pockets.

\begin{figure}[h]
\centering
\includegraphics[scale=.3]{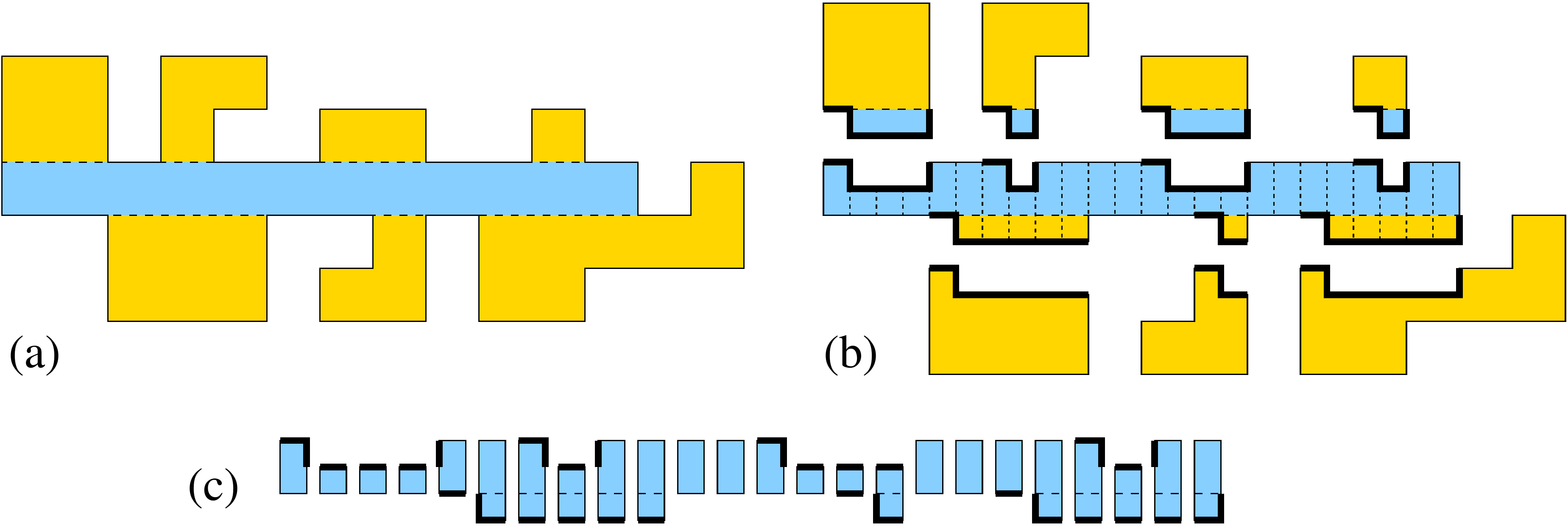}
\iffull \xxx{switch to scale2 but add a,b,c} \fi
\caption{Constructing a horizontal strip in a factor-$2$ scaled shape (a),
  augmented by jigsaw tabs and pockets to attach to adjacent pieces (b),
  proceeding column-by-column (c).}
\label{figure:scale2}
\end{figure}

  The horizontal edges of each tab or pocket uses a pair of glues.
  The unit-length upper horizontal edge uses one glue,
  and the possibly longer lower horizontal edge uses the other glue.
  The pockets at the top of strip $s$ use a different glue pair
  from the tabs at the bottom of strip~$s$.
  Furthermore, the pocket or tab connecting $s$ to its parent
  uses a different glue pair from all other pockets and tabs.
  Thus, there are four different glue pairs (for a total of eight glues).
  If the depth of $s$ in the rooted tree of the strip graph
  is even, then we use the first glue pair for the top pockets,
  the second glue pair for the bottom tabs, except for the connection
  to the parent which uses either the third or fourth glue pair
  depending on whether the connection is a top pocket or a bottom tab.
  If the depth of $s$ is odd, then we reverse the roles of the
  first two glue pairs with the last two glue pairs.
  All vertical edges of tabs and pockets use the same glue,~$8$.

  To construct the strip $s$ augmented by tabs and pockets,
  we proceed sequentially from left to right,
  as shown in Figure~\ref{figure:scale2}(c).
  The construction uses two bins.
  At the $k$th step, the primary bin contains the first $k-1$ columns
  of the augmented strip.
  In the secondary bin, we construct the $k$th column by brute force
  in one stage using 1--3 tiles and 0--2 distinct internal glues
  plus the desired glues on the boundary.
  Because the column specifies only two glues for horizontal edges,
  at the top and bottom, we can use any two other glues for the internal glues.
  All of the vertical edges of the column use different glues.
  If $k$ is odd, the left edges use glues $1$--$3$
  and the right edges uses glues $4$--$6$, according to $y$ coordinate;
  if $k$ is even, the roles are reversed.
  (In particular, these glues do not conflict with glue~$8$
  in the tabs and pockets.)
  The only exception is the first and last columns, which have no glues on
  their left and right sides, respectively.
  Now we can add the secondary bin to the primary bin,
  and the $k$th column will uniquely attach to the right side of the
  first $k-1$ columns.  In the end, we obtain the augmented strip.

  During the building of the strip, we attach children subtrees.
  Specifically, once we assemble the rightmost column of an attachment
  to one or two children strips, we recursively assemble those one or
  two children subtrees in separate bins, and then mix them into $s$'s
  primary bin.  Because the glues on the top and bottom sides of $s$
  differ, as do the glues of $s$'s parent, and because of the jigsaw approach,
  each child we add has a unique place to attach.
  Therefore we uniquely assemble $s$'s subtree.
  Applying this construction to the root of the tree, we obtain a unique
  assembly of the entire shape.
\end{proof}


\subsection{Simulation of One-Stage Assembly with Logarithmic Scale Factor}
\label{sec:simulation}


In this section, we show how to use a small number of stages to
combine a constant number of tile types into a collection of
supertiles that can simulate the assembly of an arbitrary set of
tiles at temperature $\tau = 1$, given that these tiles only
assemble fully connected shapes. 

\begin{theorem} \label{thm:simulation}
Consider an arbitrary single stage, single bin tile system with tile
set $T$, all glues of strength at most~$1$, and that assembles a
class of fully connected shapes. There is a temperature-$1$ staged
assembly system that simulates the one-stage assembly of $T$ up to
an $O(\log |T|)$ size scale factor using $3$ glues, $O(1)$ tiles,
$O(|T|)$ bins, and $O(\log\log |T|)$ stages. At the cost of
increasing temperature to $\tau=2$, the construction achieves full
connectivity. \xxx{Do bin and stage bounds hold in terms of number
of glues in input,
  instead of number of tiles?}
\end{theorem}


\xxx{need to add $3$-glue argument}

\begin{proof}
Suppose the $T$ uses $c$ distinct glue types.  As described in
Figure~\ref{figure:MacroGlue}, the initial stage of assembly
can use three distinct tile types that assemble into a supertile
representing $0$ in a first bin, and three tile types for the assembly of
a supertile representing $1$ in a second bin.  We can then split these
supertiles into four groups and attach tile types $a$ and $A$ as shown
in Figure~\ref{figure:MacroGlue}.  The third stage mixes all
possible combinations of supertiles attached to tile type $a$ with
those attached with type $A$ to get a distinct supertile for each
possible $4$-bit binary string.  This process can be repeated to
obtain all possible length-$8$ bit strings, and so on.  Thus, within
$O(\log\log c)$ stages we can obtain at least $c$ distinct binary
strings of length at most $O(\log c)$.

\begin{figure}[h]
\centering
\subfigure[]{\includegraphics[scale=.55]{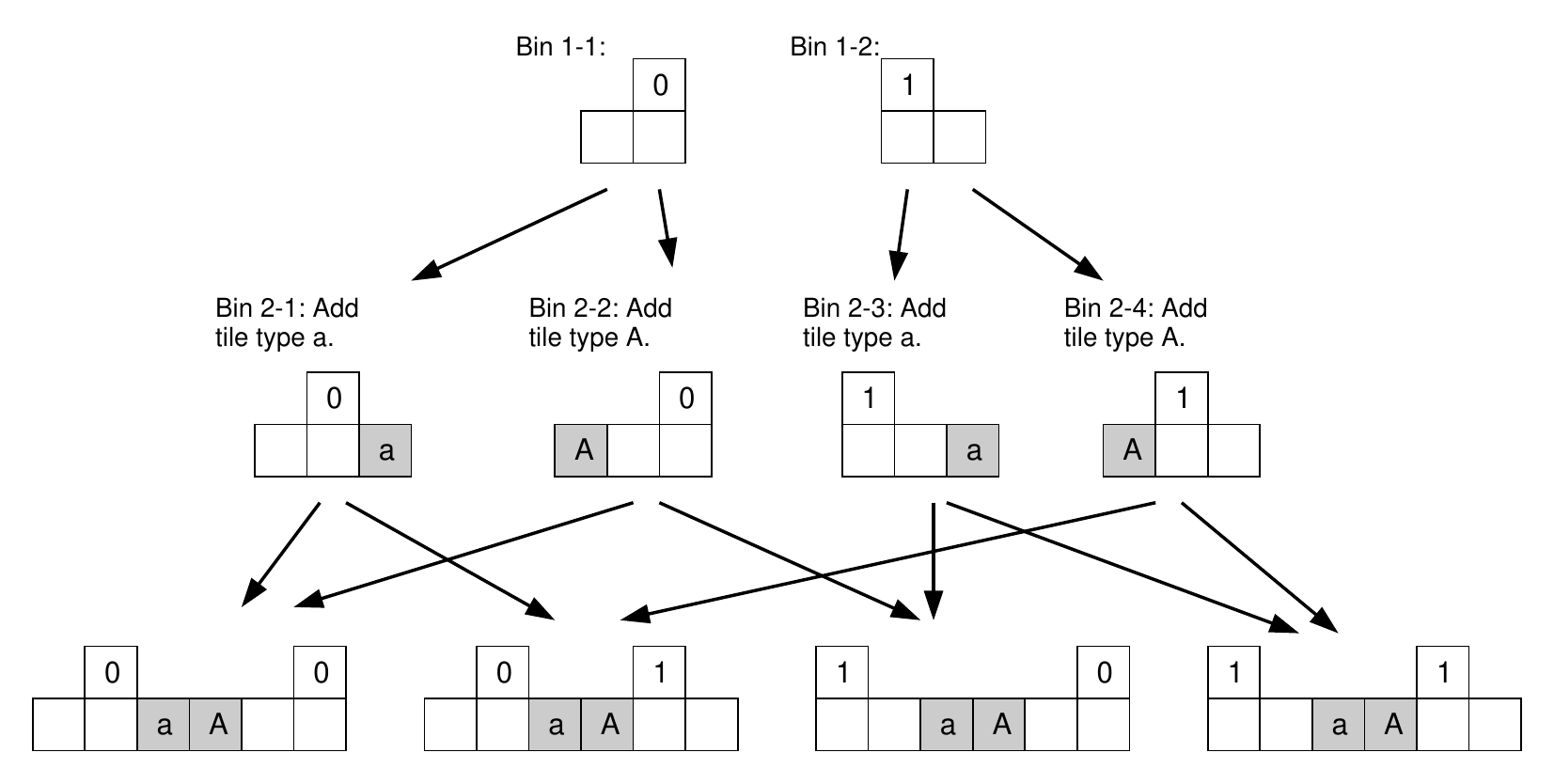}}\hfil
\subfigure[]{\includegraphics[scale=.55]{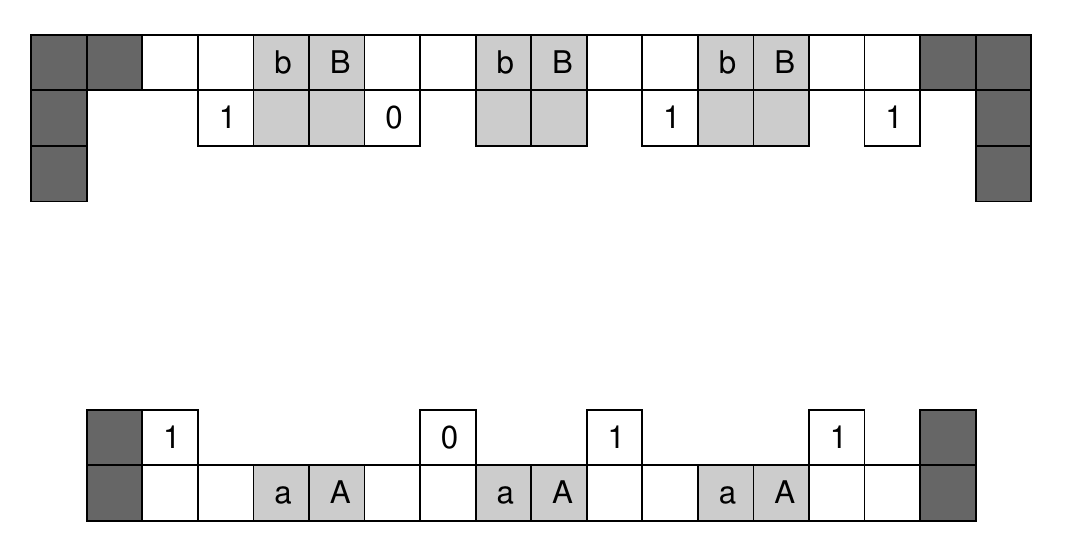}}
\caption{(a) Using $O(1)$ tile types and $O(\log r)$ stages, we can
assemble $2^r$ different supertiles, each encoding a distinct $r$
bit binary string.  (b) By creating two versions of each string and
appending tiles to the ends we can enforce that identical strings
combine while distinct strings do not.  Note that even if a single
bit differs between two strings, the rigid geometry of the
supertiles ensure that no tiles will be able to bond.}
\label{figure:MacroGlue}
\end{figure}

\begin{figure}[h]
\centering
\includegraphics[scale=.75]{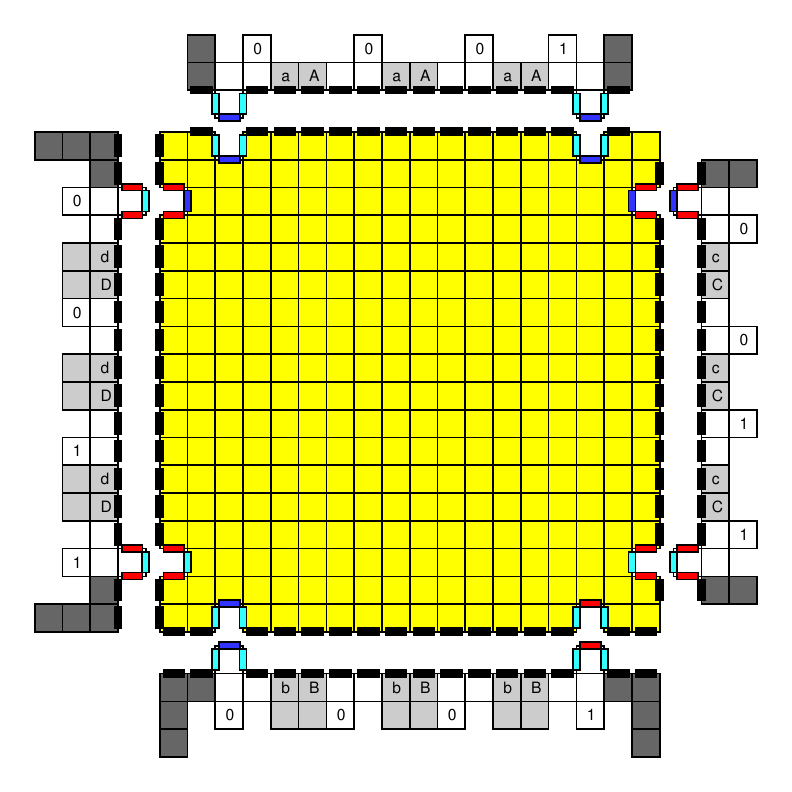}
\caption{By constructing an alphabet of binary strings for each of
the four possible tile sides, arbitrary combinations of four can be
brought together to assemble macro tiles.  This permits the
simulation of $\tau=1$ tile systems with macro blocks using only
$O(1)$ tile types.} \label{figure:TileSimulation}
\end{figure}

Repeating this process four times produces an alphabet of
glue types for each tile side.  As shown in
Figure~\ref{figure:MacroGlue}, we can make the geometry of identical
bits for opposite directions (north/south, east/west) be
interlocking.  Thus, when two glues are lined up against each other,
if all bits match, the two supertiles can lock together and get a
full bonding.  However, due to the interlocking geometry, if even a
single bit does not match, this mismatch will prevent the two
supertiles from getting close enough to get even a single bond.
Further, to prevent shifting of strings that share
prefixes/suffixes, we can attach the interlocking dark tiles shown
in Figure~\ref{figure:MacroGlue}(b).

Finally, given the four alphabets of glues with each glue type in a
separate bin, we can bring together arbitrary combinations of four
to create macro tiles as shown in
Figure~\ref{figure:TileSimulation}.  We can thus create a set of
macro tiles that will bond in the same fashion as any given target
$\tau=1$ tile system.  The holes in the constructed shape can
trivially be filled in in a nonplanar fashion by adding in a
constant size set of filler tiles.

Note that the construction does not work for simulating $\tau=2$
systems if we restrict ourselves to a constant bounded temperature.
This is because a single glue match for a macro tile yields a
\emph{large}, nonconstant number of bonds.  Further, note that when
a macro tile attaches at a position adjacent to two or more already
attached macro tiles, it cannot attach within the plane, making the
construction inherently nonplanar.
\end{proof}

Extending this simulation to
temperature-$2$ one-stage systems is an open problem.


\subsection{Assembly of Monotone Shapes}
\label{sec:monotone}

\begin{theorem}
  Any monotone shape has a fully connected temperature-$1$ staged assembly
  using $9$ glues, $O(1)$ tiles, $O(\log n)$ stages, and $O(n)$ bins,
  where $n$ is the side length of the smallest square bounding~$S$.
\end{theorem}

\begin{proof}
We assume wlog that the shape is x-monotone, which means its intersection with any vertical line is connected. We use the similar construction that we used for building square. We first decompose the shape horizontally to get long thin supertiles which we already know how to build. Here we will only discuss horizontal decomposition.
During horizontal decomposition, the challenge is to decompose a supertile $S$ into a left and a right supertile that can be combined uniquely. We decompose $S$ horizontally only when the number of columns in $S$ is greater than $3$, otherwise, we just need vertical decomposition. Let $i$, $i+1$, and $i+2$ be the three columns roughly in the middle of the supertile $S$. Column $i$ is adjacent to the column $i+1$ at certain locations. Since the shape is x-monotone, the tiles in column $i$ adjacent to column $i+1$ form a connected component. Same is the case with tiles in column $i+2$ that are adjacent to column $i+1$.

\begin{figure}[h]
\centering
\includegraphics[scale=.50]{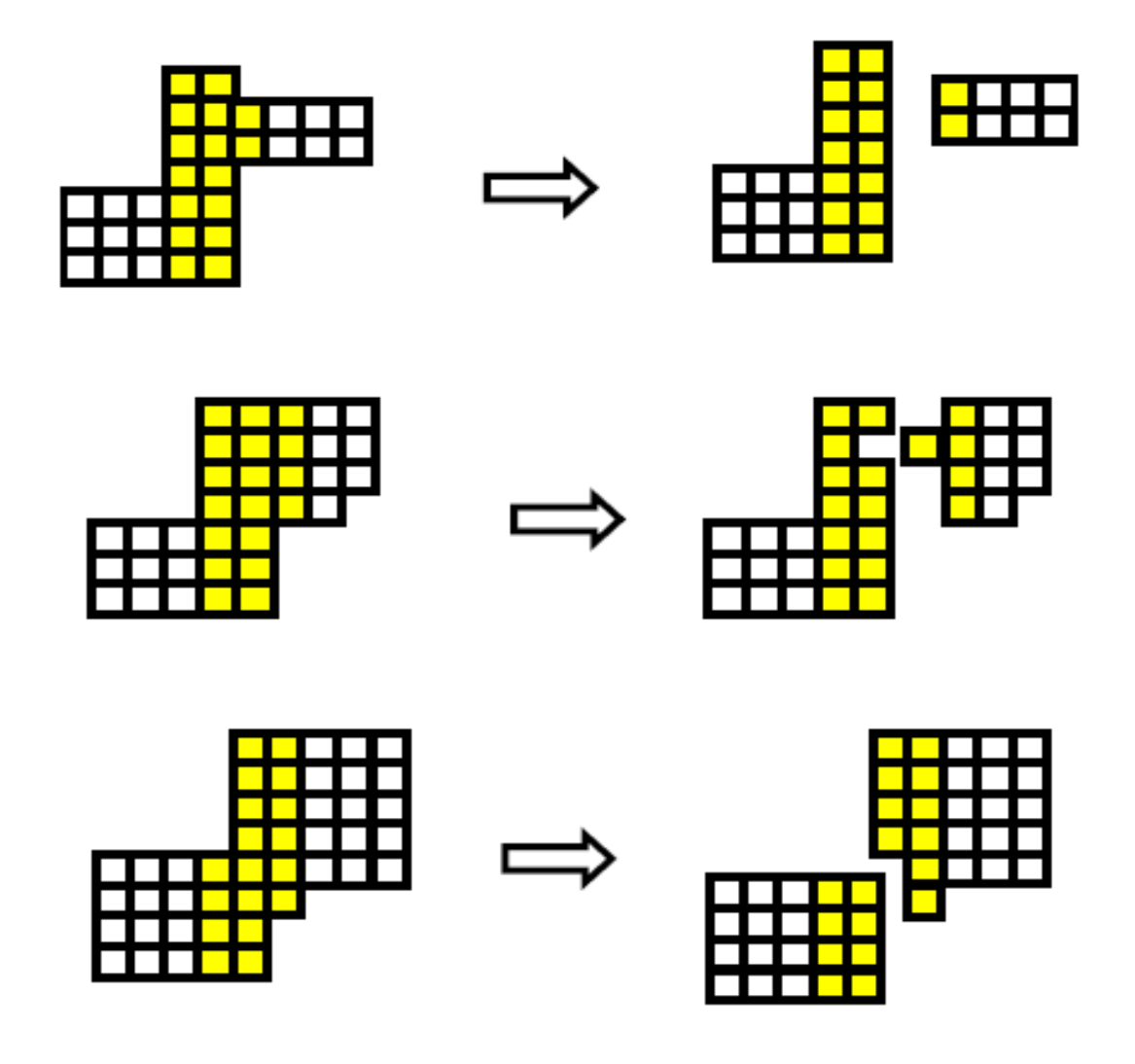}
\caption{Assembling a monotone shape.}
\label{figure:ortho}
\end{figure}

If the number of adjacent tiles between column $i$ and $i+1$ is $\leq 3$, we simply cut $S$ between column $i$ and $i+1$. Otherwise if the number of tiles in column $i+2$ adjacent to the tiles in column $i+1$ is $\leq 3$, we can break $S$ between columns $i+1$ and $i+2$. See Figure~\ref{figure:ortho} (top).

If column $i+1$ is adjacent to both columns in more than three tiles, we find the tiles in column $i+1$ that are adjacent to both columns. These tiles form a connected component due to monotonicity. If the number of such tiles $\geq 3$ we can create a jigsaw tab/pocket combination at column $i+1$. See Figure~\ref{figure:ortho} (middle). Notice the left supertile is not monotone anymore because of the last column. But we can ignore the last column because it will never be one of the three middle columns until the supertile contains only three columns and at that point we don't need horizontal decomposition any more.

If number of tiles in column $i+1$ that are adjacent to both columns is $< 3$, we decompose $S$ by creating an elbow see Figure~\ref{figure:ortho} (below). To create an elbow: assume without loss of generality that the highest tile in column $i$ adjacent to column $i+1$ is lower than the highest tile in column $i+2$ adjacent to column $i+1$. We cut the column $i+1$ such that the tiles in the column that are either adjacent to column $i$ or below any such tile belong to the left supertile and the rest of the column belong to right supertile.

The horizontal decomposition uses only constant number of only 9 glues thus $O(1)$ tiles. The decomposition tree is balanced so we need only $O(\log n)$ stages. The number of bins required can be $O(n)$ because we may need to keep each column in a separate bin.
\end{proof}



\section{Fast Counters at Temperature
$\tau=1$}\label{sec:fastCounters}
\begin{figure}[h]
\centering
\includegraphics[scale=1.0]{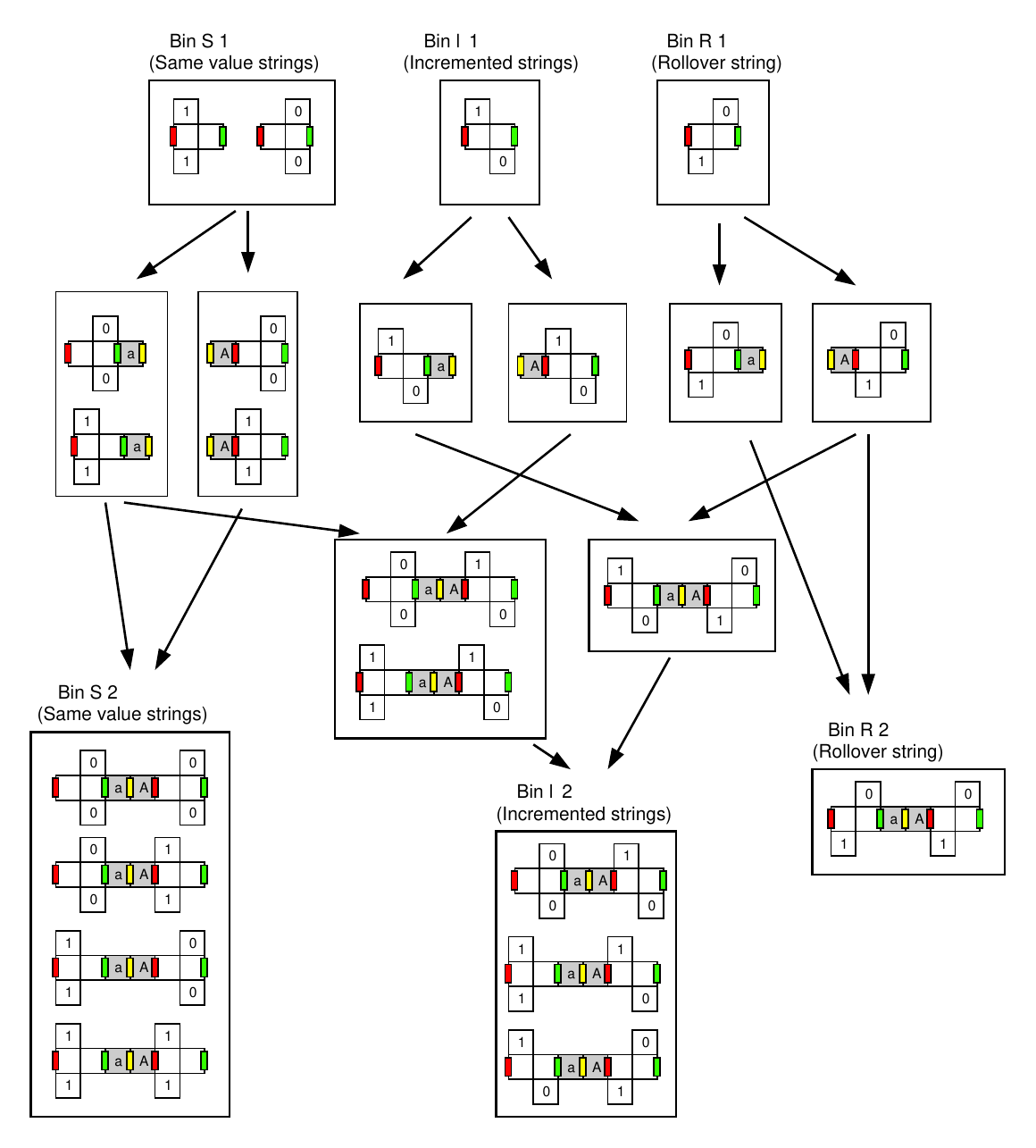}
\caption{Supertiles capable of binary counting can be constructed
efficiently by a simple recursive mixing algorithm.  A set of binary
strings of length $x$ can be assembled in $O(1)$ bin complexity and
$O(\log x)$ stage complexity.} \label{figure:counterTiles}
\end{figure}

\begin{figure}[h]
\centering
\includegraphics[scale=1.0]{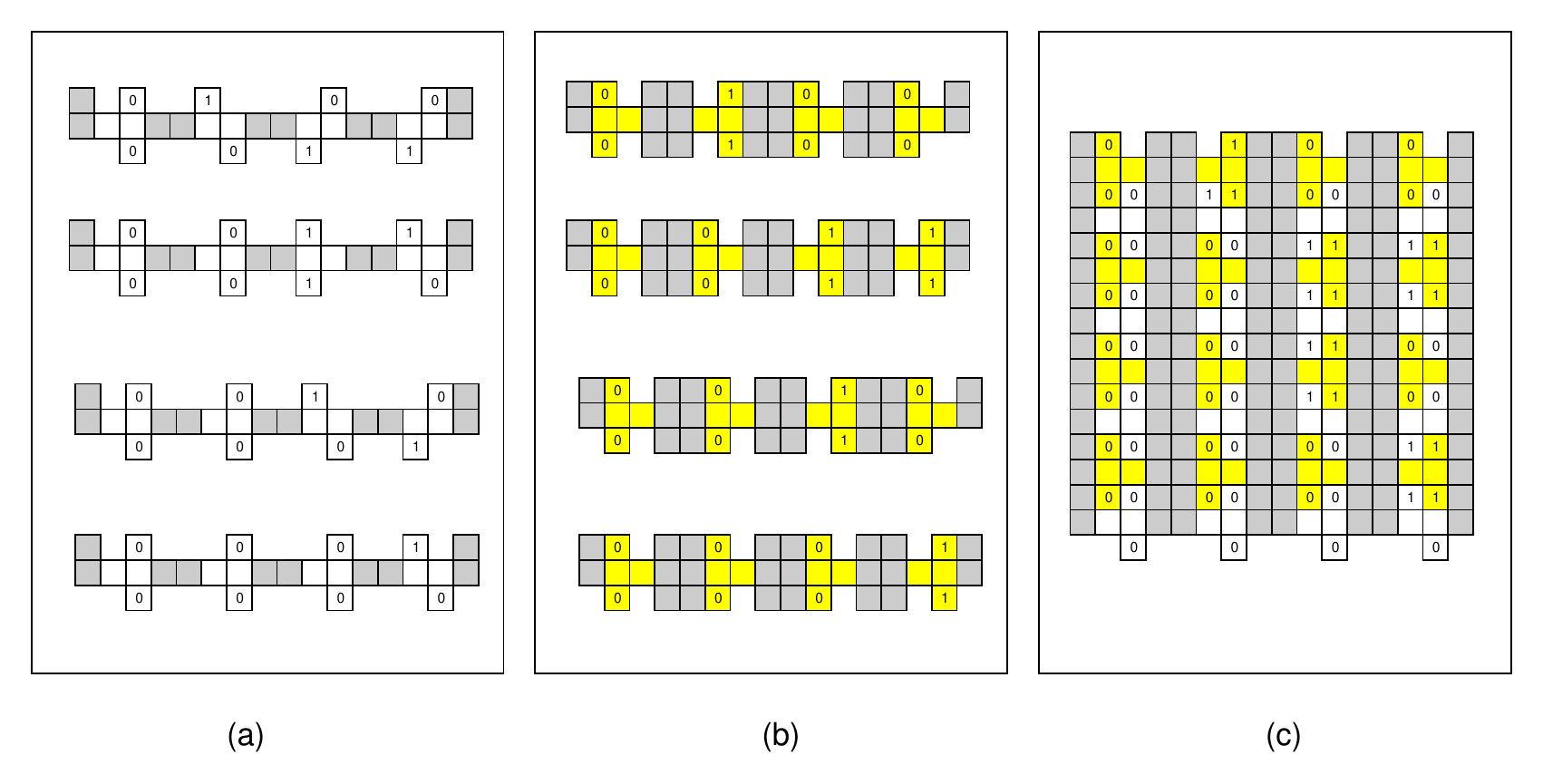}
\caption{((a) and (b)) With $O(1)$ tile complexity, $O(1)$ bin
complexity, and $O(k)$ stage complexity, two separate batches of
supertiles can be created, each containing $2^{2^k}$ distinct
supertiles. (c) When combined, supertiles may attach together by
alternating between supertiles from each group.  Further, attachment
is only possible between supertiles whose binary strings denote
values that are of difference exactly one.  The effect is thus an
assembly whose bit pattern encoded row by row represents a counter
incrementing by one until the maximum value is reached, yielding a
length $O(2^{2^k})$ assembly. In this example of length 4 strings,
only 4 of the possible 16 supertiles are shown.  With this
construction, in contrast to single stage assemblies, two successive
counter values may attach independent of whether or not previous
values have attached. Thus, the resultant structure should assemble
much quicker than other methods in which each row of a counter must
be added in succession, starting from an initial seed row.}
\label{figure:counterTilesExample}
\end{figure}

One of the most powerful and prevalent tools in the algorithmic
self-assembly literature is the
counter~\cite{Rothemund-Winfree-2000,Adleman-Cheng-Goel-Huang-2001,Aggarwal-Cheng-Goldwasser-Kao-Espanes-Moisset-Schweller-2005,Kao-Schweller-2006,Winfree-1998,Barish-Rothemund-Winfree-2005}.
A set of tiles that implement a counter are tiles that assemble into
a pattern such that successive positive integer positions are
encoded into successive positions in the assembled shape.  Such
constructions will typically then control the length of the
assembled shape by stopping growth when the counter reaches its
maximum value.  In this section we introduce a new method of
building counters in the tile assembly model that takes advantage of
the power of staged assembly.  We argue that our approach yields
some important benefits in terms of assembly speed and temperature
$\tau =1$.  Given the proven utility of counter assemblies, we
provide our construction as a primitive tool that may be useful in
the development of more efficient assembly systems.

The most typical example of a counter consists of a tile set where
each tile type is conceptually assigned either a '0' or a '1' binary
label.  For some specified value $k$, such a system assembles a $k
\times 2^k$ rectangle such that for any row $i$ in the assembly, the
$k$ tiles in the row $i$ encode the binary value of $i$ by their
assigned labels.

Counters under the standard single stage model suffer from two
drawbacks.  First, they require temperature $\tau = 2$ to work.
Second, all the constructions to date in the literature are designed
so that the $i^{th}$ value of the counter cannot attach/assemble
until the 1st through $i-1$ values have already assembled.  This
creates a lower bound of $\Omega(n)$ assembly time for these
constructions (see~\cite{Adleman-Cheng-Goel-Huang-2001} for a
definition of assembly time under the standard model).

In our construction of a binary counter, we attempt to improve upon both of
these drawbacks.  First, our construction utilizes temperature $\tau
= 1$. Second, the construction may assembly in a parallel manner.  That is, the
supertile encoding the value $i$ can attach to the supertile
encoding the value $i+1$ at any time, regardless of whether or not
the supertile representing the value $i-1$ has attached to anything.
While a definition of assembly time under the two-handed assembly
model has not yet been developed, it is plausible that this parallelism could yield a substantial reduction in assembly time for a reasonable model.

\subsection{Counter Construction}\label{sec:counterConst}

To implement the staged assembly binary counter, we design a mixing
algorithm to yield two batches of supertiles as shown in
Figure~\ref{figure:counterTilesExample}, each including a list of
long thin supertiles encoding a bit pattern of interlocking teeth on
the north and south surface of the supertile in the same fashion as
Theorem~\ref{thm:simulation}. In particular, the first batch will
consist of supertiles whose pattern of interlocking teeth on the
north face of the supertile encode the binary string obtained by
incrementing the binary string encoded on the south face of the
supertile by one.  The second batch is similar, but the string
encoded on the north and south face of each supertile is not
incremented.

By design, the glues on the north and south faces of each supertile
in either batch are distinct, making attachment among supertiles
impossible.  However, we can make the north glues used in the first
batch the same as the south glues used in the second batch, and vice
versa.  From this we get that when the two batches are mixed
together supertiles may assemble by alternating between supertiles
from the first batch and supertiles from the second batch.

Further, due to the geometry encoded on the surface of each
supertile, each supertile attaches above a supertile whose binary
value is exactly one less than its own. Thus, any assembled
structure consists of a chain of rows, each row representing an
incremented binary value.  Therefore the unique terminal assembly is
such that the northmost face is a supertile encoding the highest
value string of all 1's, while the south face consists of the
supertile representing the string of all 0's.

To see how to assemble the binary strings used in this construction,
consider the problem of assembling a set of supertiles such that
each of the $2^{2^i}$ length $2^i$ binary strings is represented by
a supertile encoding the string on its south surface.  Further, for
each such supertile in the set, require that it encodes the value
encoded on its south surface incremented by 1 on its north surface
(assume the all 1's string incremented is the all 0's string).
Denote this set as $X_i$.  Such a set is essentially the first batch
of Figure~\ref{figure:counterTilesExample}, and a straightforward
modification of the following technique can yield the second batch
as well.

Now, to obtain a bin whose unique assemblies are $X_i$, it is
sufficient to obtain 2 bins whose assemblies union are equal to
$X_i$, as these bins can be combined within 1 stage.  Let $I_i$ ($I$
for incremented strings) denote the subset of strings in $X_i$ minus
the string whose south surface is all 1's.  Let $R_i$ ($R$ for
rollover strings) simply be the supertile encoding all 1's on the
south face and all 0's on the north face.  Finally, define a third
set $S_i$ not contained in $X_i$, where $S_i$ is the set of all
length $i$ strings encoding the same value on the south and north
face of the supertile.

To describe how to attain a bin with the set $X_i$ as uniquely
produced supertiles, we show how to recursively compute the three
sets $S_i$, $I_i$, and $R_i$.  Assume, as depicted in
Figure~\ref{figure:counterTiles}, that each supertile in the sets
$S$, $I$, and $R$ must have a strength 1 red and green glue on the
west most and east most center edge respectively.

Recursively, assume we already have 3 separate bins containing
$S_{i/2}$, $I_{i/2}$, and $R_{i/2}$.  Within a single stage, split
the contents of each of these three bins into 2 separate bins (for a
total of six distinct bins).  Denote the bins by $S^a_{i/2}$ and
$S^A{i/2}$ etc.  For the $a$ bins, add tile $a$ from
Figure~\ref{figure:counterTiles}. For the $A$ bins, add tile $A$.

Now combine sets $S^a_{i/2}$ and $S^A_{i/2}$.  This yields a bin
containing the set of all length $i$ binary strings that have the
same values on the north and south faces, which is the set $S_i$.

Now combine set $S^a_{i/2}$ and $I^A_{i/2}$.  This yields a set of
supertiles that is a subset of $I_i$, namely the strings (encoded on
the south face of the supertile) whose least significant 0 occurs in
the right half of the string.  The remaining set of $I_i$, the
strings whose least significant 0 occurs in the left half of the
string, is obtained by combining $I^a_{i/2}$ and $R^A_{i/2}$. A
third stage thus yields the set $I_i$.

Finally, the set $R_i$ is obtained by combining $R^a_{i/2}$ and
$R^A_{i/2}$.

As base case for this recursive mixing procedure, we can build the
sets for $i=1$ by brute force with distinct tile types.  This
technique uses at most 6 bins and 3 stages per recursion level.
Thus, the desired set $X_x$ can be obtained in $O(1)$ bins and
$O(\log x)$ stages.  The procedure for extending size 1 strings to
size 2 strings is depicted in Figure~\ref{figure:counterTiles}.

\subsection{Counting up to general $n$}
The counter described in Section~\ref{sec:counterConst} counts from
value 0 up to $2^{2^k} -1$ for a specified value $k$ using $k$
stages, $O(1)$ bins, and $O(1)$ tile complexity. However, this
construction clearly is not immediately capable of assembling
supertiles of arbitrary length $n$.  In contrast, constructions
exist at $\tau=2$ under the single stage model such that the exact
length of counters can be specified.  This can typically be done by
specifying an initial first value of the counter as these systems
always start from a seed value. However, with our approach this is
much more difficult.

Currently, we have a complex construction combining the technique of
Theorem~\ref{thm:crazyupper} with the binary counter system of
Section~\ref{sec:counterConst} to yield unique assembly of a counter
of any length $n$ at temperature $\tau =1$, $O(B)$ bin complexity,
and $O(\frac{\log n}{B^2} + \log B)$ stage complexity.  However, we
do not include the details of this construction as it is very
complex and as of yet does not have direct application to building
shapes of interest, such as squares.  However, we conjecture that
this technique can yield a square building scheme that improves
Theorem\ref{thm:crazyupper} to a $\tau =1$ construction.

\section{Future Directions}
\label{sec:summary}

There are several open research questions stemming from this work.

One direction is to relax the assumption that, at each stage,
all supertiles self-assemble to completion.  In practice, it is
likely that at least some tiles will fail to reach their terminal
assembly before the start of the next stage.  Can a staged assembly
be robust against such errors, or at least detect these errors
by some filtering, or can we bound the error propagation in some
probabilistic model?

Another direction is to develop a model of the assembly time
required by a mixing operation involving two bins of tiles.
Such models exist for (one-stage) \emph{seeded self-assembly}---which
starts with a seed tile and places singleton tiles one at a time---but
this model fails to capture the more parallel nature of two-handed assembly
in which large supertiles can bond together without a seed.
Another interesting direction would be to consider nondeterministic assembly
in which a tile system is capable of building a large class of distinct
shapes.  Is it possible to design the system so that certain shapes
are assembled with high probability?

Another research direction is the consideration of 3D assembly.  We have focused on two-dimensional constructions in this paper
which provides a more direct comparison with previous models,
and 
 is also a case of practical interest, e.g., for manufacturing sieves.
Many of our results, in theory, also generalize to 3D (or any constant dimension),
at the cost of increasing the number of glues and tiles.
For example, the spanning-tree model generalizes trivially, and a modification
to the jigsaw idea enables many of the other results to carry over.  However, 3D assembly in practice is much harder than 2D assembly, stemming in part from the fact that 2D assembly systems in practice make use of 3 dimensions.  How to properly model and address the difficulties of 3D assembly is an important research direction.  In particular, combining our staged assembly techniques with existing error correcting mechanisms seems a potentially fruitful direction for further research.

Finally, experimental validation of our model and techniques is an extremely important direction for future work.  It is likely that simulations and implementations of staged assembly techniques will yield key insights into the model, providing a road map for future work.



\paragraph{Acknowledgments.}

We thank M. S. AtKisson and Edward Goldberg for extensive discussions
about the bioengineering application.


\let\realbibitem=\bibitem
\def\bibitem{\par \vspace{-1.2ex}\realbibitem}

\begin{small}
\bibliography{selfassembly,tiling}
\bibliographystyle{alpha}
\end{small}


\end{document}